\documentclass{article}

\usepackage[final]{neurips_2024}

\usepackage[utf8]{inputenc} 
\usepackage[T1]{fontenc}    
\usepackage{hyperref}
\hypersetup{breaklinks=true,colorlinks,urlcolor={red!90!black},linkcolor={red!90!black},citecolor={blue!90!black},pagebackref=true,bookmarks=false}
\usepackage{url}            
\usepackage{booktabs}       

\usepackage{amsfonts}       
\usepackage{nicefrac}       
\usepackage{microtype}      
\usepackage{xcolor}         

\usepackage{algorithm}
\usepackage{algpseudocode}
\usepackage{graphicx}
\usepackage{amsthm}
\usepackage{amsmath}
\usepackage{xspace}
\usepackage{cleveref}
\usepackage{wrapfig}

\newtheorem{theorem}{Theorem}[section]

\newtheorem{lemma}[theorem]{Lemma}

\newcommand{\method}{\textsc{AliDiff}\xspace}

\usepackage{amsmath,amsfonts,bm}

\def\eqref#1{equation~\ref{#1}}

\def\1{\bm{1}}

\def\rr{{\textnormal{r}}}

\def\rvf{{\mathbf{f}}}

\def\rvh{{\mathbf{h}}}

\def\rvm{{\mathbf{m}}}

\def\rvp{{\mathbf{p}}}

\def\rvr{{\mathbf{r}}}
\def\rvs{{\mathbf{s}}}

\def\rvv{{\mathbf{v}}}

\def\rvx{{\mathbf{x}}}

\def\vmu{{\bm{\mu}}}
\def\vtheta{{\bm{\theta}}}

\def\vc{{\bm{c}}}

\def\mI{{\bm{I}}}

\DeclareMathAlphabet{\mathsfit}{\encodingdefault}{\sfdefault}{m}{sl}
\SetMathAlphabet{\mathsfit}{bold}{\encodingdefault}{\sfdefault}{bx}{n}

\def\gD{{\mathcal{D}}}

\def\gL{{\mathcal{L}}}
\def\gM{{\mathcal{M}}}

\def\gP{{\mathcal{P}}}

\newcommand{\R}{\mathbb{R}}

\makeatletter
\DeclareRobustCommand\onedot{\futurelet\@let@token\@onedot}
\def\@onedot{\ifx\@let@token.\else.\null\fi\xspace}

\def\eg{\textit{e.g}\onedot}

\makeatother

\title{Aligning Target-Aware Molecule Diffusion Models\\with Exact Energy Optimization}

\author{%
  Siyi Gu$^{* 1}$, Minkai Xu$^{* 1 \dagger}$, Alexander Powers$^1$, Weili Nie$^2$, Tomas Geffner$^2$\\
  \textbf{Karsten Kreis$^2$, Jure Leskovec$^1$, Arash Vahdat$^2$, Stefano Ermon$^1$}\\
  $^{1}$ Stanford Univeristy $^{2}$ NVIDIA\\
  \texttt{\{sgu33,minkai,jure,ermon\}@cs.stanford.edu \, lxpowers@stanford.edu}\\
  \texttt{\{wnie,tgeffner,kkreis,avahdat\}@nvidia.com}
}

\begin{document}

\maketitle

\begin{abstract}
\looseness=-1
    Generating ligand molecules for specific protein targets, known as structure-based drug design, is a fundamental problem in therapeutics development and biological discovery. Recently, target-aware generative models, especially diffusion models, have shown great promise in modeling protein-ligand interactions and generating candidate drugs. However, existing models primarily focus on learning the chemical distribution of all drug candidates, which lacks effective steerability on the chemical quality of model generations. In this paper, we propose a novel and general alignment framework to align pretrained target diffusion models with preferred functional properties, named \method.
    \method shifts the target-conditioned chemical distribution towards regions with higher binding affinity and structural rationality, specified by user-defined reward functions, via the preference optimization approach. To avoid the overfitting problem in common preference optimization objectives, we further develop an improved Exact Energy Preference Optimization method to yield an exact and efficient alignment of the diffusion models, and provide the closed-form expression for the converged distribution. Empirical studies on the CrossDocked2020 benchmark show that \method can generate molecules with state-of-the-art binding energies with up to -7.07 Avg.~Vina Score, while maintaining strong molecular properties. 
    Code is available at \url{https://github.com/MinkaiXu/AliDiff}.
\end{abstract}

\section{Introduction}

\renewcommand{\thefootnote}{\fnsymbol{footnote}}
\footnotetext[1]{Equal contribution; junior author listed earlier. $^\dagger$Correspondence to: Minkai Xu <{minkai@cs.stanford.edu}>.}
\renewcommand{\thefootnote}{\arabic{footnote}}

Generating ligand molecules with desirable properties and high affinity to protein targets, known as structure-based drug design (SBDD), is a fundamental problem in therapeutic design and biological discovery. 
It necessitates methods that can produce realistic and diverse drug-like molecules with stable 3D structures and high binding affinities. 
The problem is challenging due to the vast combinatorial chemical space~\citep{ragoza2022generating} and high degree of freedom of protein targets~\citep{qiao2024state}. In the past few years, numerous deep generative models have been proposed  to generate molecules in string sequences~\citep{kusner2017grammar,segler2018generating} or atom-bond
graph representations~\citep{jin2018junction,Shi2020GraphAF}. Although these models have shown promise in generating plausible drug-like molecules, they lack sufficient modeling of the 3D protein-ligand interaction with proteins and therefore can hardly be adopted in target-aware molecule generation. As a result, generating ligands conditioned on protein targets remains an open problem.

Recently, with rapid progress in structural biology and the increasing scale of structural data~\citep{francoeur2020three,jumper2021highly}, numerous target-aware generative models have been proposed to directly generate molecules within the protein targets in 3D. Initial work proposed to sequentially place atoms within the target via autoregressive models~\citep{luo20213d, liu2022generating, peng2022pocket2mol}, while later work learns diffusion models to jointly design the whole ligand with state-of-the-art results~\citep{guan20233d, lin2022diffbp, schneuing2022structure,huang2023protein, guan2024decompdiff}. Following the biological principle to model the protein-ligand complex interactions, these methods have shown great promise in generating realistic drugs that can bind toward given targets.
However, all existing models solely focus on learning the chemical distribution of candidate molecules and treat all training samples equally, while in practice, only the ligand molecules with strong binding affinity and high synthesizability are preferred for real-world therapeutic development. As a result, existing learned models generally lack sufficient steerability regarding the relative quality of model generations and cannot generate faithful samples with the desirable properties. 

\begin{wrapfigure}{r}{0.5\textwidth}
\vspace{-5pt}
    \includegraphics[width=0.99\linewidth]{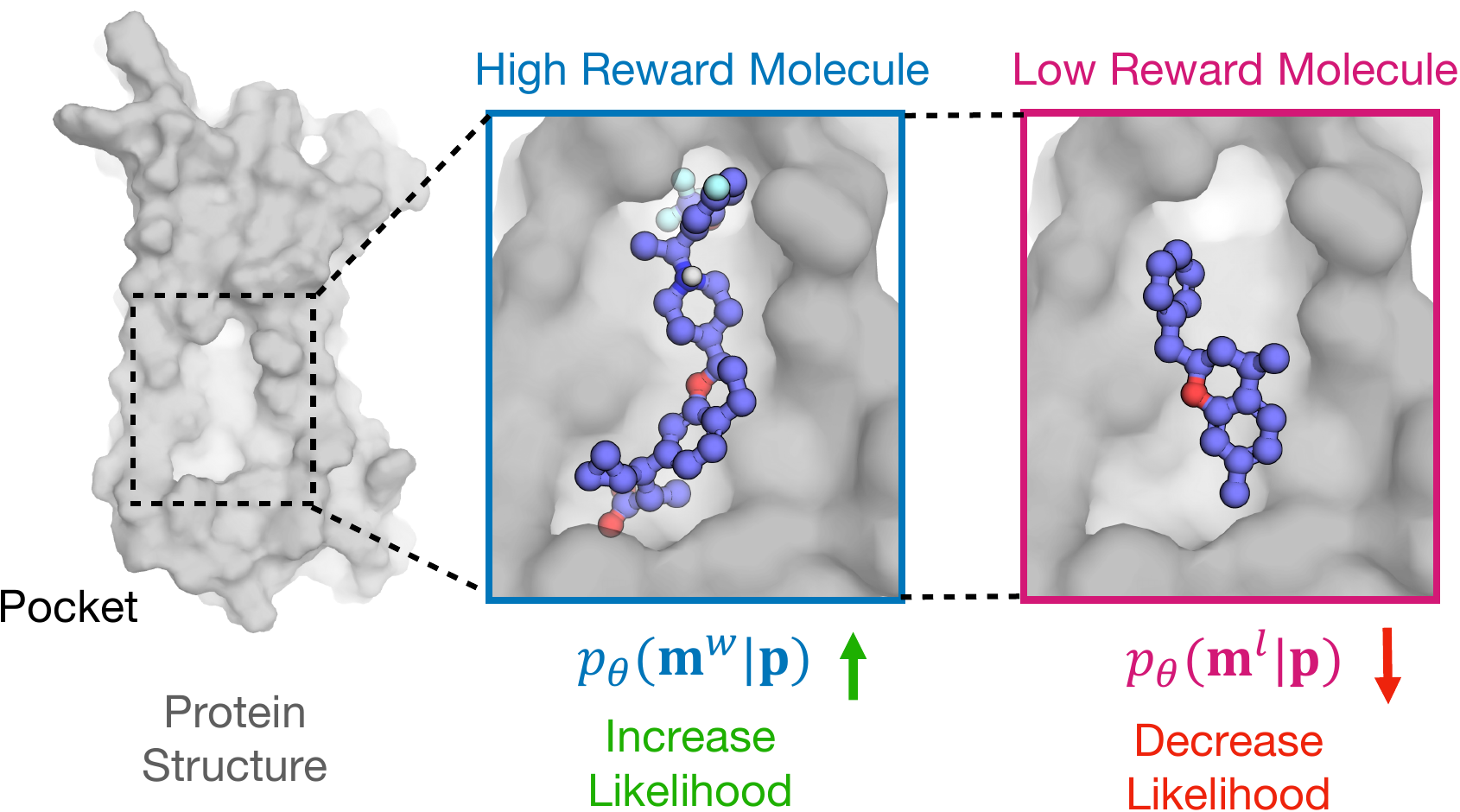} 
    \caption{High-level illustration of \method. For a protein target, we can have multiple candidate ligands and rank the preference by certain reward functions, \eg, binding energy. We align the target-aware molecule diffusion model with these preferences by adjusting the conditional likelihoods.}
\label{fig:wrapfig}
\end{wrapfigure}

\looseness=-1
To bridge the gap between existing SBDD models and the necessity for designing ligands with favorable properties,  in this paper, we introduce a novel and comprehensive alignment framework to align pretrained target-aware diffusion models with preferred functional properties, named \method. \method adjusts the target-conditioned chemical distribution toward regions characterized by lower binding energy and structural rationality, as specified by a user-defined reward function, using a preference optimization approach. To this end, we derive a unified variational lower bound to align the likelihoods of both discrete chemical type and continuous 3D coordinate features. We further analyze the winning data overfitting problem commonly associated with preference optimization objectives, and introduce an improved Exact Energy Preference Optimization (E$^2$PO) method. E$^2$PO analytically ensures a precise and efficient alignment of diffusion models, and we provide a closed-form expression for the converged distribution. 
Our key contributions can be summarized as follows:
\begin{itemize}
    \item We address the challenge of designing favorable target-aware molecules from the perspective of aligning molecule generative models with desirable properties. We introduce the energy preference optimization framework and derive variational lower bounds to align diffusion models for generating molecules with high binding affinity to binding targets.
    \item We analyze the overfitting issue in the preference optimization objective, and propose an improved exact energy optimization method to yield an exact alignment towards target distribution shifted by reward functions.
    \item We conduct comprehensive comparisons and ablation studies on the CrossDocked2020~\citep{francoeur2020three} benchmark to justify the effectiveness of \method. Empirical results demonstrate that \method can generate molecules with state-of-the-art binding energies with up to -7.07 Avg. Vina Score, while maintaining strong molecular properties.
\end{itemize}

\section{Related Work}

\textbf{Structure-Based Drug Design.} With increasing amount of structural data becoming accessible, generative models have attracted growing attention for structure-based molecule generation. 
Early research~\citep{skalic2019target} proposes to generate SMILES representations from protein contexts by sequence generative models. 
Inspired by the progress in 3D and geometric modeling, many works proposed to solve the problem directly in 3D space. 
For instance, \citet{ragoza2022generating} voxelizes molecules within atomic density grids and generates them through a Variational Autoencoder framework. \citet{luo20213d,peng2022pocket2mol, powers2024} developed autoregressive models to generate molecules by sequentially placing atoms or chemical groups within the target. Following the autogressive backbone, FLAG\citep{zhang2023molecule} and DrugGPS~\citep{zhang2023learning} take advantage of chemical priors of molecular fragments to generate ligand molecules piece by piece, leading to more realistic substructures. More recently, 
diffusion models achieved exceptional results in synthesizing high-quality images and texts, which have also been successfully used for ligand molecule generation~\citep{guan20233d, lin2022diffbp, schneuing2022structure,huang2023protein, guan2024decompdiff}. These models generate molecules by progressively denoising atom types and coordinates while maintaining physical symmetries with SE(3)-equivariant neural networks. 
While the existing works focus on designing molecules using various deep generative models, they often struggle with generating molecules that exhibit different desirable properties, \eg, strong binding affinity, high synthesizability, and low toxicity. Real-word drug discovery projects almost always seek to optimize or constrain these properties~\citep{d2012multi,bickerton2012quantifying}. 
In this work, we aim to address the challenge with a novel and general preference optimization framework.

\textbf{Reinforcement learning from human feedback (RLHF).}  Recently, significant efforts have been devoted to aligning generative models with human preferences. The use of reinforcement learning to incorporate feedback from humans and AI into finetuning large language models is exemplified by Reinforcement Learning from Human Feedback (RLHF)~\citep{ziegler2020finetuning,ouyang2022training}. Research works have incorporated human feedback to improve performance across various domains, such as machine translation~\citep{nguyen2017reinforcement}, summarization~\citep{stiennon2020learning}, and also diffusion models~\citep{uehara2024feedback,uehara2024fine}. Notably, \citet{rafailov2023direct} designed a new preference paradigm that enables training language models to satisfy human preferences directly without reinforcement learning. This algorithm was later applied to diffusion models for text-to-image generation tasks~\citep{wallace2023diffusion}. Concurrent work~\citep{zhou2024antigen} attempts to apply DPO for designing antibodies with rationality and functionality. To the best of our knowledge, we are the first alignment approach for target-aware ligand design, where the conditional distribution is shifted toward desirable properties.

\begin{figure*}[!t]
    \centering
    \includegraphics[width=1.0\linewidth]{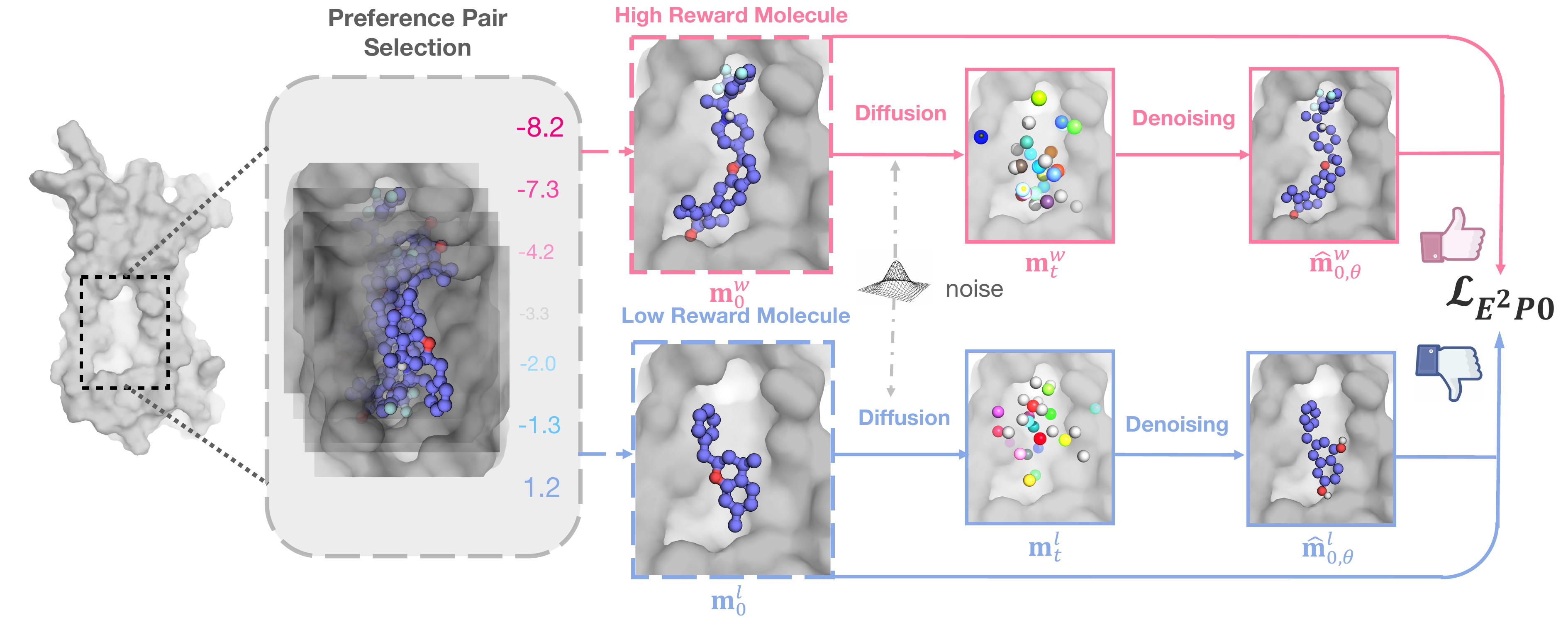}
    \vspace{-15pt}
    \caption{Overview of \method. This workflow can be summarized as 1) For each protein target (pocket) $\rvp$ in the training set, we retrieve two candidate ligands $\rvm$; 2) Label the two ligands as wining sample $\rvm^w$ and losing sample $\rvm^l$ by desirable properties, \eg, binding energies; 3) Calculate the preference optimization objective \Cref{eq:loss-alidiff-exact} and update the molecule diffusion model $p_\theta$.}
    \label{fig:framework}
\end{figure*}

\section{Method}

In this section we present \method, a general framework for aligning target-aware diffusion models with various molecular functionalities. We first provide an overview of the target-aware ligand diffusion model and our Reinforcement Learning from Feedback formulation (\cref{subsec:method-overview}). Next, we introduce the energy optimization approach for aligning the diffusion model and analyze the potential limitations of the framework (\cref{subsec:method-dpo}). We then further introduce an exact energy optimization method from a distribution matching perspective to align the generative model efficiently and exactly (\cref{subsec:method-epo}). A visualization of the framework is shown in \Cref{fig:framework}.

\subsection{Overview}
\label{subsec:method-overview}

\textbf{Notation.} We focus on aligning molecule generative models for structure-based drug design, which can be abstracted as generating molecules that can bind to a given protein target. Following the convention in the  related literature~\citep{luo20213d,guan20233d}, the molecule and target protein are represented as $\gM=\{(\rvx_{M}^{(i)}, \rvv_{M}^{(i)}) \}_{i=1}^{N_M}$ and $\gP = \{ (\rvx_{P}^{(i)}, \rvv_{P}^{(i)}) \}_{i=1}^{N_P}$, respectively, where $N_M$ and $N_P$ denote the number of atoms of the molecule $\mathcal{M}$ and the protein $\mathcal{P}$. $\rvx \in \R^3$ and $\rvv \in \R^K$ denote the atomic 3D position and chemical type, respectively, with $K$ being the dimension of atom types. 
For brevity, we denote the molecule as a  matrix $\rvm=[\rvx_M, \rvv_M]$ where $\rvx_M\in \R^{N_M\times 3}$ and $\rvv_M\in \R^{N_M\times K}$, and denote the protein as a matrix $\rvp=[\rvx_P, \rvv_P]$ where $\rvx_P\in \R^{N_P\times 3}$ and $\rvv_P\in \R^{N_P\times K}$. 
The task can then be formulated as modeling the conditional distribution $p(\rvm|\rvp)$.

\textbf{Preliminaries.} Diffusion Models have been previously used to model the joint distribution of atomic types and positions~\citep{guan20233d,schneuing2022structure,lin2022diffbp}. 
This approach consists of a forward diffusion
process and a reverse generative (denoising) process. Both processes are only defined on the ligand molecules $\rvm$, with fixed proteins $\rvp$. In the forward process, small Gaussian and categorical noises are gradually injected on atomic coordinates $\rvx$ and types $\rvv$ as follows:
\begin{align}
\label{eq:diffusion-forward-kernal}
    q(\rvm_t|\rvm_{t-1}, \rvp)=
    \mathcal{N}(\rvx_{t};\sqrt{1-\beta_t}\rvx_{t-1},\beta_t\mI) \cdot \mathcal{C}(\rvv_{t};(1-\beta_t)\rvv_{t-1}+\beta_t/K),
\end{align}
where $\mathcal{N}$ and $\mathcal{C}$ stand for the Gaussian and categorical distribution respectively, and $\beta_t$ corresponds to a  (fixed or learnable) variance schedule. Note that, in certain recent work $q$ process can be learnable with dependence on the conditioning $\rvp$~\citep{huang2023protein}.
We omit the subscript $M$ for the ligand molecule without ambiguity here and denote the atom positions and types at time step $t$ as $\rvx_t$ and $\rvv_t$.  
Using Bayes theorem, the posterior conditioned on $\rvm_0$ can be computed in closed form: 
\begin{align}
\label{eq:diffusion-forward}
    q(\rvm_{t-1}|\rvm_t,\rvm_0, \rvp)
    =\mathcal{N}(\rvx_{t-1};\Tilde{\vmu}(\rvx_{t}, \rvx_{0}),\Tilde{\beta}_t\mI) \cdot \mathcal{C}(\rvv_{t-1};\Tilde{\vc}(\rvv_{t},\rvv_{0})),
\end{align}
where 
$\Tilde{\vmu}(\rvx_{t},\rvx_{0})=\frac{\sqrt{\Bar{\alpha}_{t-1}}\beta_t}{1-\Bar{\alpha}_t}\rvx_{0}+\frac{\sqrt{\alpha}_t(1-\Bar{\alpha}_{t-1})}{1-\Bar{\alpha}_t}\rvx_{t}$, $\Tilde{\beta}_t=\frac{1-\Bar{\alpha}_{t-1}}{1-\Bar{\alpha}_t}\beta_t$, 
$\alpha_t=1-\beta_t$, 
$\Bar{\alpha}_t=\prod_{s=1}^t \alpha_s$,  
$\Tilde{\vc}(\rvv_{t},\rvv_{0})=\frac{\vc^*}{\sum_{k=1}^Kc^*_k}$, and $\vc^*(\rvv_{t},\rvv_{0})=[\alpha_t\rvv_{t}+(1-\alpha_t)/K]\odot[\Bar{\alpha}_{t-1}\rvv_{0}+(1-\Bar{\alpha}_{t-1})/K]$~\citep{ho2020denoising,austin2011d3pm}.
At timestep $T$, $q$ converges to the prior with Gaussians on coordinates and uniforms on atom types.
The reverse process, also known as the generative process, learns a neural network parameterized by $\theta$ to recover data by iterative denoising. The denoising step can be approximated with predicted Gaussians $\vmu_\theta$ and categorical distributions $\vc_\theta$ as follows:
\begin{equation}
\begin{aligned}
\label{eq:diffusion-denoising}
    p_\theta(\rvm_{t-1}|\rvm_t, \rvp) & = \mathcal{N}(\rvx_{t-1};{\vmu}_\theta ([{\rvx}_t, {\rvv}_{t}], t, \rvp),\Tilde{\beta}_t\mI)\cdot \mathcal{C}(\rvv_{t-1}; {\vc}_\theta ([{\rvx}_t, {\rvv}_{t}], t, \rvp)) \\
    & = \mathcal{N}(\rvx_{t-1};\Tilde{\vmu}(\rvx_{t},\hat{\rvx}_{0}),\Tilde{\beta}_t\mI)\cdot \mathcal{C}(\rvv_{t-1}; \Tilde{\vc}(\rvv_{t},\hat{\rvv}_{0})),
\end{aligned}
\end{equation}
where $[\hat{\rvx}_{0}, \hat{\rvv}_{0}] = \epsilon_\theta ([{\rvx}_t, {\rvv}_{t}], t, \rvp)$ are predictions from a denoising network $\epsilon_\theta$. Importantly, the denoising network here is specifically parameterized by equivariant neural networks, resulting in an SE(3)-invariant likelihood $p_\theta(\rvm|\rvp)$ on the protein-ligand complex~\citep{xu2022geodiff}.


\textbf{Overview.} As ligand molecules with desirable properties, \eg, high binding affinity and synthesizability, are required for real-world therapeutic development, we aim to align the ligand diffusion model with these preferences. 
Such preferences can be defined as a reward model $r(\cdot): \gM \times \gP \rightarrow \mathbb{R}$ 
calculated from various cheminformatics software, \eg, binding affinity, drug-likeness, synthesizability, or their combinations.
We fine-tune and align the pre-trained diffusion model with the reinforcement learning framework. Specifically, given a dataset $\gD$ containing given protein targets, inspired by RLHF~\citep{ouyang2022training}, this fine-tuning is achieved by maximizing the reward: 
\begin{align}
\label{eq:rlhf}
    \max_{p_\vtheta} \mathbb{E}_{\rvp \sim \gD, \rvm \sim p_{\vtheta}} [r(\rvm, \rvp)]-\beta\mathbb{D}_{\text{KL}}(p_\vtheta(\rvm | \rvp)\Vert p_{\text{ref}}(\rvm | \rvp)),
\end{align}
where $p_\vtheta$ and $p_{\text{ref}}$ are the distributions induced by the fine-tuned and pre-trained models, respectively. In this work, $p_\vtheta$ and $p_{\text{ref}}$ are the fine-tuned and pre-trained molecule diffusion models, as introduced above. $\beta$ is a hyperparameter controlling the KL divergence regularization.
Note that, here the reward is a known black-box function, unlike typical RLHF where it is unknown and has to be estimated from preferences.
In the following section, we elaborate on how the alignment objective is rewritten with diffusion forward and reverse processes defined on atomic types and coordinates.

\subsection{Energy Preference Optimization}
\label{subsec:method-dpo}


Though the reward function is known, evaluating reward values such as binding affinity is computationally expensive and we instead resort to aligning with a labeled offline dataset.
We start with a dataset \( \mathcal{D} = \{ (\rvp, \rvm^w, \rvm^l) \} \) where $\rvp$ denotes the protein condition and \( \rvm^w \succ \rvm^l \) is a pair of winning and losing ligands with respect to certain specified energy, \eg, binding energy. 
The optimal solution to the RLHF objective from \Cref{eq:rlhf} can be written in closed-form $p^*_\theta (\rvm|\rvp) \propto p_\textnormal{ref}(\rvm|\rvp) \exp(\frac{1}{\beta} r(\rvm, \rvp))$~\citep{peters2007reinforcement}. Following the preference optimization algorithm~\citep{rafailov2023direct}, we use the Bradley Terry (BT, \citep{bradley1952rank}) model $p(\rvm^0_1 \succ \rvm^0_2 | \rvp)=\sigma(r(\rvm^0_1, \rvp)-r(\rvm^0_2, \rvp))$ to reformulate the RLHF objective as:
\begin{equation}
\label{eq:dpo}
    \mathcal{L}_{\text{DPO}}(\theta) = -\mathbb{E}_{(\rvp, \rvm^w, \rvm^l) \sim \mathcal{D}} \left[ \log \sigma \left( \beta \log \frac{p_\theta(\rvm_0^w | \rvp)}{p_{\text{ref}}(\rvm_0^w | \rvp)} - \beta \log \frac{p_\theta(\rvm_0^l | \rvp)}{p_{\text{ref}}(\rvm_0^l | \rvp)} \right) \right].
\end{equation}
Due to the intractability of $p_\theta(\rvm|\rvp)$ for diffusion models, 
we instead follow recent work on diffusion-based preference optimization~\citep{wallace2023diffusion} to align the whole reverse process
and utilize Jensen’s inequality to optimize its negative evidence lower bound optimization (ELBO):
\begin{equation}
    \mathcal{L}_{\text{DPO-Diffusion}}(\theta) = -\mathbb{E}_{(\rvp, \rvm_0^w, \rvm_0^l) \sim \mathcal{D}, (\rvm_{1:T}^w, \rvm_{1:T}^l)\sim p_\theta} \left[ \log \sigma \left( \beta \log \frac{p_\theta(\rvm_{0:T}^w)}{p_{\text{ref}}(\rvm_{0:T}^w)} - \beta \log \frac{p_\theta(\rvm_{0:T}^l)}{p_{\text{ref}}(\rvm_{0:T}^l)} \right) \right],
\end{equation}
where we omit the conditioning on the protein target $\rvp$ for compactness. We further approximate the reverse process $ p_\theta(\rvm_{1:T}|\rvm_0) $ with the forward process $ q(\rvm_{1:T}|\rvm_0) $ for efficient sampling of $\rvm_{1:T}$, and obtain the following expression after some derivations~\citep{wallace2023diffusion}:
\begin{equation}
\begin{aligned}
\label{eq:loss-dpo-diffusion}
    \tilde{\mathcal{L}}_{\text{DPO-Diffusion}}& (\theta) = -\mathbb{E}_{(\rvp, \rvm_0^w, \rvm_0^l) \sim \mathcal{D}, t\sim[0,T], \rvm_{t}^w \sim q, \rvm_{t}^l\sim q} \big[\\
    & \log \sigma \big( -\beta T \big( \mathbb{D}_\textnormal{KL}(q(\rvm_{t-1}^w|\rvm_{0,t}^w) \| p_\theta (\rvm_{t-1}^w|\rvm_t^w)) - \mathbb{D}_\textnormal{KL}(q(\rvm_{t-1}^w|\rvm_{0,t}^w) \| p_{\text{ref}}(\rvm_{t-1}^w|\rvm_{t}^w)) \\
    & - \mathbb{D}_\textnormal{KL}(q(\rvm_{t-1}^l|\rvm_{0,t}^l) \| p_\theta (\rvm_{t-1}^l|\rvm_{t}^l)) + \mathbb{D}_\textnormal{KL}(q(\rvm_{t-1}^l|\rvm_{0,t}^l) \| p_{\text{ref}}(\rvm_{t-1}^l|\rvm_{t}^l)) \big) \big) \big]
\end{aligned}
\end{equation}
Let $[\hat{\rvx}_0, \hat{\rvv}_0]$ be the predicted atom position and type, which are fed into \Cref{eq:diffusion-denoising} to obtain the posterior distributions. With the joint diffusion processes \Cref{eq:diffusion-forward-kernal,eq:diffusion-forward,eq:diffusion-denoising} on both continuous $\rvx$ and discrete $\rvv$ features, the above KL divergences can be decomposed and calculated as:
\begin{equation}
\begin{aligned}
    & \mathbb{D}_\textnormal{KL}(q(\rvm_{t-1}|\rvm_{0, t}) \| p(\rvm_{t-1}|\rvm_{t})) = \mathbb{D}_\textnormal{KL}^{\rvx,t-1}(q(\rvx_{t-1}|\rvx_{0, t}) \| p(\rvx_{t-1}|\rvx_{t})) + \mathbb{D}_\textnormal{KL}^{\rvv,t-1}(q(\bm{c}_{t-1}|\bm{c}_{0, t}) \| p(\bm{c}_{t-1}|\bm{c}_{t})), \\
    & \mathbb{D}_\textnormal{KL}^{\rvx,t-1}(q(\rvx_{t-1}|\rvx_{0, t}) \| p(\rvx_{t-1}|\rvx_{t})) = \frac{1}{\Tilde{\beta_t}} \| \Tilde{\bm\mu}(\rvx_t, \rvx_0) - \Tilde{\bm\mu}(\rvx_t, \hat{\rvx}_0)\|^2 + C = \gamma_t \|\rvx_0 - \hat\rvx_0 \|^2 + C,\\
    & \mathbb{D}_\textnormal{KL}^{\rvv,t-1}(q(\bm{c}_{t-1}|\bm{c}_{0, t}) \| p(\bm{c}_{t-1}|\bm{c}_{t})) = \sum_{k} \Tilde{\bm{c}}(\rvv_t, \rvv_0)_k \log \frac{\Tilde{\bm{c}}(\rvv_t, \rvv_0)_k}{\Tilde{\bm{c}}(\rvv_t, \hat\rvv_0)_k},
\end{aligned}
\end{equation}
where $\gamma_t = \frac{\bar\alpha_{t-1} \beta_t^2}{2\sigma_t^2(1 - \bar\alpha_t)^2}$ and $C$ is a constant. Let $\hat{\rvx}_{0,\theta}, \hat{\rvv}_{0,\theta} $ and $\hat{\rvx}_{0,\textnormal{ref}}, \hat{\rvv}_{0,\textnormal{ref}}$ be the predictions from the fine-tuned and from the original pretrained model, respectively. Then, we can further obtain the preference optimization loss on $\rvx$ and $\rvv$, respectively, as follows:
\begin{align}
    & \mathcal{L}_{t-1}^\rvx(\theta) = -\mathbb{E} \big[ \log \sigma \big( -\beta T \gamma_t ( ||\rvx_0^w - \hat{\rvx}_{0,\theta}^w||^2 - ||\rvx_0^w - \hat{\rvx}_{0,\textnormal{ref}}^w||^2  - ||\rvx_0^l - \hat{\rvx}_{0,\theta}^l||^2 + ||\rvx_0^l - \hat{\rvx}_{0,\textnormal{ref}}^l||^2 ) \big) \big] \notag\\
    & \mathcal{L}_{t-1}^\rvv(\theta) = -\mathbb{E} \big[ \log \sigma \big( -\beta T \big( \mathbb{D}_\textnormal{KL}(\Tilde{\bm{c}}(\rvv_t^w, \rvv_0^w) || \Tilde{\bm{c}}(\rvv_t^w, \hat{\rvv}^w_{0,\theta})) - \mathbb{D}_\textnormal{KL}(\Tilde{\bm{c}}(\rvv_t^w, \rvv_0^w) || \Tilde{\bm{c}}(\rvv_t^w, \hat{\rvv}^w_{0,\textnormal{ref}})) \nonumber \\
    & \quad \quad \quad \quad \quad \quad - \mathbb{D}_\textnormal{KL}(\Tilde{\bm{c}}(\rvv_t^l, \rvv_0^l) || \Tilde{\bm{c}}(\rvv_t^l, \hat{\rvv}_{0,\theta}^l)) + \mathbb{D}_\textnormal{KL}(\Tilde{\bm{c}}(\rvv_t^l, \rvv_0^l) || \Tilde{\bm{c}}(\rvv_t^l, \hat{\rvv}^l_{0,\textnormal{ref}})) \big) \big) \big]
\end{align}
With Jensen's inequality and the convexity of $-\log \sigma$, we can derive the final objective as a (weighted) sum of atom coordinate and type preference losses $\gL_{t-1}^{\rvx} + \gL_{t-1}^{\rvv}$, which turns the sum of the KL terms outside $-\log \sigma$ and serves as an upper bound of \Cref{eq:loss-dpo-diffusion}:
\begin{equation}
\label{eq:loss-alidiff}
    \mathcal{L}_\method(\theta) =  -\mathbb{E}_{(\rvp, \rvm_0^w, \rvm_0^l) \sim \mathcal{D}, t\sim[0,T], \rvm_{t}^w \sim q, \rvm_{t}^l\sim q} \big[ \gL_{t-1}^{\rvx} + \gL_{t-1}^{\rvv} \big] \geq \tilde{\mathcal{L}}_{\text{DPO-Diffusion}}(\theta),
\end{equation}
where the preference is assigned separately to atom types $\rvv$ and coordinates $\rvx$. The loss decomposition imposes a fine-grained preference assignment on chemical elements and geometric strcutures and enables us to choose weights to balance the training of the two variables~\citep{guan20233d,guan2024decompdiff}.
The overall training and sampling algorithms of \method are summarized in \Cref{app:sec:alg}. 


\subsection{Exact Energy Optimization}
\label{subsec:method-epo}


Although DPO enjoys the advantage of efficient fine-tuning without fitting a reward function, recent theoretical investigations reveal that it is highly vulnerable to overfitting by pushing all the probability mass on the winning sample~\citep{azar2024general}.
Specifically, the non-linear transformation $\log \sigma$ of \Cref{eq:dpo} pushes the $\log p_\theta(\rvm^w|\rvp) - \log p_\theta(\rvm^l|\rvp)$ towards infinity, completely removing the likelihood for the losing sample regardless of any regularization in the original RLHF setup \Cref{eq:rlhf}~\citep{azar2024general,tang2024generalized}.
Let us analyze the problem with an example consisting of two ligand molecules $\rvm^w$ and $\rvm^l$ with their rewards measured as $\rvr^w$ and $\rvr^l$ (\eg, calculated from binding energy).
The DPO objective in \Cref{eq:loss-alidiff} tends to just greedily maximize towards $p(\rvm^w \succ \rvm^l | \rvp) \rightarrow 1$. However, the optimal preference probability can be calculated by the BT model~\citep{bradley1952rank} as $\hat{p}(\rvm^w \succ \rvm^l | \rvp)=\sigma(\rvr^w-\rr^l)$, and our alignment goal is to shift the distribution to align with this $\hat{p}$ instead of greedy maximization. 
To address the over-optimization issue, we introduce an improved objective with regularization on preference maximization, named Exact Energy Preference Optimization (E$^2$PO).
Let $\bar{\mathcal{L}}_t^\rvx(\theta)$ and $\bar{\mathcal{L}}_t^\rvv(\theta)$ denote terms for reverse preference optimization:
\begin{align}
    \bar{\mathcal{L}}_{t-1}^\rvx(\theta) = 1 - {\mathcal{L}}_{t-1}^\rvx(\theta), \quad \bar{\mathcal{L}}_{t-1}^\rvv(\theta) = 1-{\mathcal{L}}_{t-1}^\rvv(\theta).
\end{align}
Our E$^2$PO objective function takes a cross-entropy form to align the distributions $p_\theta(\rvm^w \succ \rvm^l | \rvp)$ towards $\hat{p}(\rvm^w \succ \rvm^l | \rvp)$. 
Formally, it is given by:
\begin{equation}
\begin{aligned}
\label{eq:loss-alidiff-exact}
    \mathcal{L}_\text{\method-E$^2$PO}(\theta) =  -&\mathbb{E}_{(\rvp, \rvm_0^w, \rvm_0^l) \sim \mathcal{D}, t\sim[0,T], \rvm_{t}^w \sim q, \rvm_{t}^l\sim q} \big[ \\
    &\quad \quad (\sigma (\rvr^w -\rvr^l)) (\gL_{t-1}^{\rvx} + \gL_{t-1}^{\rvv} ) + (1-\sigma (\rvr^w -\rvr^l)) (\bar{\gL}_{t-1}^{\rvx} + \bar{\gL}_{t-1}^{\rvv} ) \big],
\end{aligned}
\end{equation}
where the second term $\bar{\gL}_{t-1}^{\rvx} + \bar{\gL}_{t-1}^{\rvv}$ weighted by $1 - \sigma (\rvr^w -\rvr^l)$ helps to alleviate the overfitting on the winning data sample. Notably, for $\rvr^w >> \rvr^l$, we have $\sigma(\rvr^w - \rvr^l)\approx 1$, indicating that the regularized objective in \Cref{eq:loss-alidiff-exact} will still change back to the original objective in \Cref{eq:loss-alidiff}, where overfitting on the extremely better data is expected. In principle, with the regularization objective, we have:
\begin{theorem}
\label{theorem-kl}
    The objective function in \Cref{eq:loss-alidiff-exact} optimizes a variational upper bound of the KL-divergence $\mathbb{D}_\textnormal{KL}\big(\hat{p}^*(\rvm|\rvp)||\hat{p}_\theta(\rvm|\rvp)\big)$, where $\hat{p}^*(\rvm|\rvp) \propto p_\textnormal{ref}(\rvm|\rvp) \exp(r(\rvm, \rvp))$ and $\hat{p}_\theta(\rvm|\rvp) \propto p_\textnormal{ref} (\rvm|\rvp) \left(\frac{p_\theta (\rvm|\rvp)}{p_\textnormal{ref}(\rvm|\rvp)}\right)^\beta $.
\end{theorem}
The theorem provides an analytical guarantee for the optimal shifted distribution after alignment that avoids over-optimization. Assuming we achieve convergence on the KL divergence, we have that $p_\textnormal{ref}(\rvm|\rvp) \exp(r(\rvm, \rvp)) \propto p_\textnormal{ref}^{1-\beta} (\rvm|\rvp) p_\theta^\beta (\rvm|\rvp) $ which further gives us $ p_\theta (\rvm|\rvp) \propto p_\textnormal{ref}(\rvm|\rvp) \exp(\frac{1}{\beta}r(\rvm, \rvp))$, where a smaller $\beta$ encourages a sharper shift towards the user-defined reward function. 
We give the full derivations in \Cref{app:sec:proof}, and analyze the empirical effect on generation quality in \Cref{sec:exp}.

\section{Experiment}
\label{sec:exp}

\subsection{Experiment Setup}
\textbf{Dataset}. We train and evaluate \method using the CrossDocked2020 dataset~\citep{francoeur2020three}. Following the common setup in this field~\citep{luo20213d,guan20233d}, we refined the initial 22.5 million docked protein binding complexes by selecting docking poses with RMSD lower than 1Å with the ground truth and diversifying proteins with a sequence identity below 30\%. To apply \method, we further preprocess our data and construct a dataset of the form $\gD = \{(\rvp, \rvm^w, \rvm^l)\}$, where $\rvp$ denotes the protein, $\rvm{^w}$ denotes the preferred molecules, and $\rvm^{l}$ denotes rejected molecules based on the user-defined reward. In our setting, we choose two ligand molecules per pocket site and label the preference by a certain reward, \eg binding energy for our main benchmark. We provide ablations with more reward functions in \Cref{subsec:exp-ablation}.
Details of preference pair selection are presented in \Cref{app:sec:ablation}. The final dataset uses a train and test split of 65K and 100. 

\textbf{Baselines}. We compare our model with the following baselines: liGAN~\citep{ragoza2022generating} is a conditional VAE model that utilizes a 3D CNN architecture to both encode and generate voxelized representations of atomic densities; AR~\citep{luo20213d}, Pocket2Mol~\citep{peng2022pocket2mol} and GraphBP~\citep{liu2022generating} are autoregressive models that learn graph neural networks to generate 3D molecules atom by atom sequentially; TargetDiff~\citep{guan20233d} and DecompDiff~\citep{guan2024decompdiff} are diffusion-based approaches for generating atomic coordinates and types via a joint denoising process; IPDiff~\citep{huang2023protein} is the most recent state-of-the-art diffusion-based approach that further integrates the interactions between the target protein and the molecular ligand into the generation process. 

\textbf{Evaluation metrics.} We evaluate the generated molecules by comparing \textit{binding affinity} with the target and critical \textit{molecular properties}. We analyze the generated molecules across 100 test proteins, reporting the mean and median for affinity-based metrics (Vina Score, Vina Min, Vina Dock, and High Affinity) and molecular property metrics (drug-likeness QED~\citep{bickerton2012quantifying}, synthesizability SA~\citep{ertl2009estimation}, and diversity). 
We use AutoDock Vina~\citep{eberhardt2021autodock} to estimate binding affinity scores, using the common setup described by \citet{luo20213d,ragoza2022generating}.
Specifically, Vina Score estimates binding affinity from the generated 3D structures, Vina Min refines the structure through local minimization before estimation, Vina Dock uses a re-docking procedure to reflect the optimal binding affinity, and High Affinity gauges the percentage of generated molecules that bind better than reference molecules per protein.


\subsection{Results}

\begin{table}[!t]
\caption{Summary of binding affinity and molecular properties of reference molecules and molecules generated by \method and baselines. (↑) / (↓) denotes whether a larger / smaller number is preferred. Top 2 results are bolded and underlined, respectively.}
\label{table:main}
\centering
{\resizebox{\textwidth}{!}{
\begin{tabular}{l|cc|cc|cc|cc|cc|cc|cc}
\toprule
\textbf{Methods} & \multicolumn{2}{c|}{\textbf{Vina Score (↓)}} & \multicolumn{2}{c|}{\textbf{Vina Min (↓)}} & \multicolumn{2}{c|}{\textbf{Vina Dock (↓)}} & \multicolumn{2}{c|}{\textbf{High Affinity(↑)}}   & \multicolumn{2}{c|}{\textbf{QED(↑)}}  & \multicolumn{2}{c|}{\textbf{SA(↑)}} & \multicolumn{2}{c}{\textbf{Diversity(↑)}} \\
& Avg. & Med. & Avg. & Med. & Avg. & Med. & Avg. & Med. & Avg. & Med. & Avg. & Med. & Avg. & Med.\\
\midrule
liGAN* & - & - & - & - & -6.33 & -6.20 & 21.1\% & 11.1\%  & 0.39 & 0.39 & 0.59 & 0.57 & 0.66 & 0.67 \\
GraphBP* & - & - & - & - & -4.80 & -4.70 & 14.2\% & 6.7\% & 0.43 & 0.45 & 0.49 & 0.48 & \textbf{0.79} & \textbf{0.78} \\
AR & -5.75 & -5.64 & -6.18 & -5.88 & -6.75 & -6.62  & 37.9\% & 31.0\%  & 0.51 & 0.50 & \underline{0.63} & \underline{0.63} &0.70 & 0.70 \\
Pocket2Mol & -5.14 & -4.70 & -6.42 & -5.82 & -7.15 & -6.79 & 48.4\% & 51.0\% & \textbf{0.56} & \textbf{0.57} & \textbf{0.74} & \textbf{0.75} & 0.69 & 0.71 \\
TargetDiff & -5.47 & -6.30 & -6.64 & -6.83 & -7.80 & -7.91 & 58.1\% & 59.1\% & 0.48 & 0.48 & 0.58 & 0.58 & 0.72 & 0.71 \\
DecompDiff & -5.67 & -6.04 & -7.04 & -7.09 & -8.39 & -8.43 & 64.4\% & 71.0\% & 0.45 & 0.43 & 0.61 & 0.60 & 0.68 & 0.68 \\
IPDiff & \underline{-6.42} & \underline{-7.01} & \underline{-7.45} & \underline{-7.48} & \underline{-8.57} & \underline{-8.51} & \underline{69.5\%} & \underline{75.5\%} & \underline{0.52} & \underline{0.53} & 0.61 & 0.59 & \underline{0.74} & \underline{0.73} \\
\textbf{\method} & \textbf{-7.07} & \textbf{-7.95} & \textbf{-8.09} & \textbf{-8.17} & \textbf{-8.90} & \textbf{-8.81}  & \textbf{73.4\% } & \textbf{81.4\%} & 0.50 & 0.50 & 0.57 & 0.56 & 0.73 & 0.71 \\
Reference & -6.36 & -6.41 & -6.71 & -6.49 & -7.45 & -7.26 & - & - & 0.48 & 0.47 & 0.73 & 0.74 & - & - \\
\bottomrule
\end{tabular}
}}
\end{table}

\begin{figure*}[!t]
    \centering
    \includegraphics[width=\linewidth]{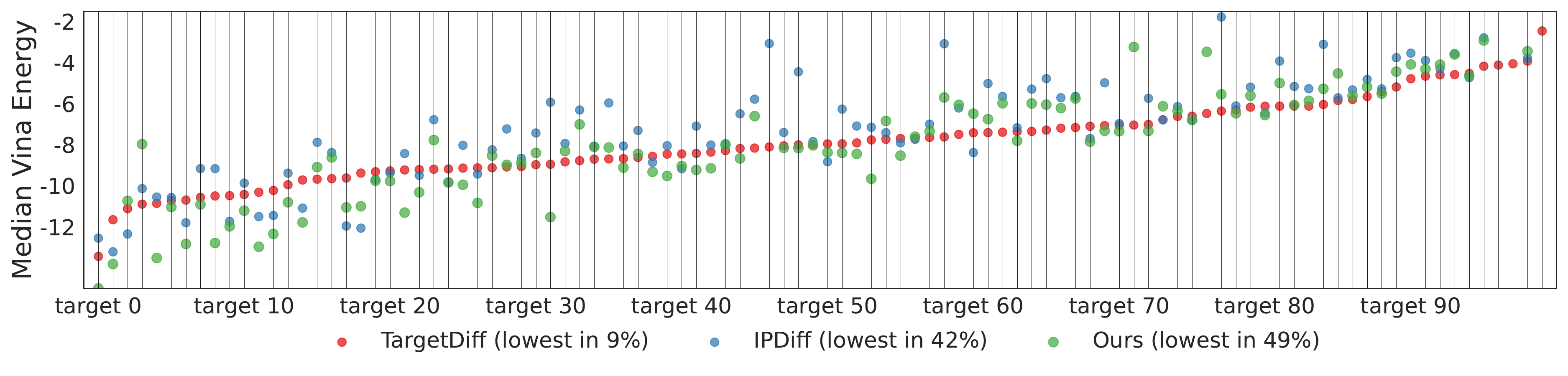}
     \vspace{-20pt}
    \caption{Median Vina energy for different generated molecules (TargetDiff, IPDiff, \method) across 100
testing samples, sorted by the median Vina energy of molecules generated from \method.
}
    \label{fig:binding}
\end{figure*}

\begin{figure*}[!b]
    \centering
    \includegraphics[width=\linewidth]{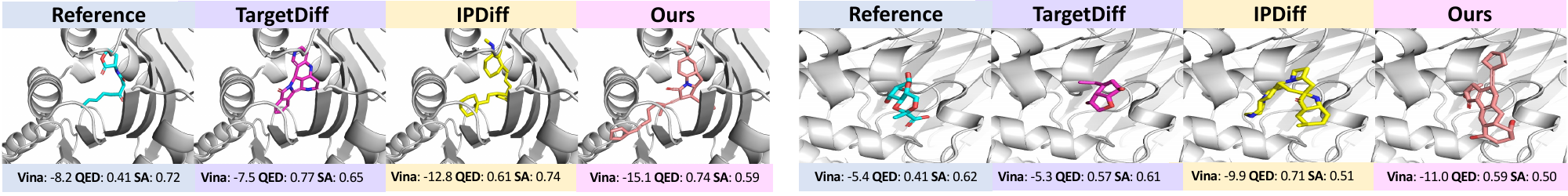}
    \vspace{-10pt}
    \caption{Visualizations of reference molecules and generated ligands for protein pockets (\texttt{1l3l}, \texttt{2e24}) generated by TargetDiff, IPDiff, and \method. Vina score, QED, and SA are reported below. }
    \label{fig:sample_vis}
\end{figure*}

\textbf{Binding Affinity and Molecular Properties.}
We compare the performance of our proposed method \method against the above baseline methods. Our model is fine-tuned from IPDiff, the ligand generative model. We report the results in \Cref{table:main}, and leave more implementation details in \Cref{app:sec:implementation}. As shown in the results, \method significantly outperforms all non-diffusion-based models in binding-related metrics, and also surpasses our base model IPDiff in all binding affinity related metrics by a notable margin.  
In particular, \method increases the binding-related
metrics Avg.~Vina Score, Vina Min, and Vina Dock by 10.1\%, 8.56\%, and 3.9\% compared with IPDiff. Our superior performance in binding-related metrics demonstrates the effectiveness of energy preference optimization. \Cref{fig:binding} shows the median Vina energy of the proposed model, compared with TargetDiff and IPDiff, two diffusion-based state-of-the-art models in target-aware molecule generation. We observe that \method surpasses these baseline models and generates molecules with the highest binding affinity for 49\% of the protein targets in the test set. In property-related metrics, we observe only a slight decrease in QED, SA, and diversity, compared with IPDiff. Specifically, with approximately 10.1\% improvement on Avg.~Vina Score, we observe a minor decrease in Avg.~SA (-6.5\%), Avg.~QED (-3.8\%), and diversity(-1.4\%). Figure~\ref{fig:sample_vis} presents examples of ligand molecules generated by \method, TargetDiff, and IPDiff. The figure shows that our generated molecules maintain reasonable structures and high binding affinity compared with all baselines, indicating their potential as promising candidate ligands. Additional experimental results and visualized examples of these molecules are in \Cref{app:sec:ablation,app:sec:vis}. 
 
We also notice a trade-off between binding affinity and property-related metrics. While we achieve state-of-the-art performance on all binding affinity metrics, the performance on QED and SA metrics slightly decreases. This phenomenon has been commonly observed in previous studies where achieving high binding affinity can often sacrifice other molecular metrics~\citep{guan20233d,huang2023protein}. This is because the highest affinity can potentially only be achieved by rather specific and unique molecules, which are harder to synthesize than simple molecules, and hence these trade-offs are expected. Besides, in real-world drug discovery, binding affinity is typically a more critical metric as molecules with more stable interaction with the pocket site are important, whereas QED and SA work mainly as rough filters~\citep{guan20233d}. For these reasons, we believe the deterioration in molecular properties is well compensated by the improvement in binding affinity, especially with such little deterioration in property metrics. 
In addition, in the following section (\Cref{table:exp:multiple}), we further discuss incorporating molecular properties into the reward, which shows slightly lower performance gain on affinity but archives improvements also on molecular properties.

\subsection{Ablation Studies}
\label{subsec:exp-ablation}

\begin{table}[!t]
\caption{Effect of combining multiple reward objectives. Affinity denotes \method, whereas Affinity+SA denotes combining both synthetic accessibility and affinity as reward function. }
\label{table:exp:multiple}
\centering
{\resizebox{\textwidth}{!}{
\begin{tabular}{l|cc|cc|cc|cc|cc|cc|cc}
\toprule
\textbf{Choice of reward} & \multicolumn{2}{c|}{\textbf{Vina Score (↓)}} & \multicolumn{2}{c|}{\textbf{Vina Min (↓)}} & \multicolumn{2}{c|}{\textbf{Vina Dock (↓)}} & \multicolumn{2}{c|}{\textbf{High Affinity(↑)}}   & \multicolumn{2}{c|}{\textbf{QED(↑)}}  & \multicolumn{2}{c|}{\textbf{SA(↑)}} & \multicolumn{2}{c}{\textbf{Diversity(↑)}} \\
& Avg. & Med. & Avg. & Med. & Avg. & Med. & Avg. & Med. & Avg. & Med. & Avg. & Med. & Avg. & Med.\\
\midrule
Affinity & -7.07 & -7.95 & \textbf{-8.09} & \textbf{-8.17} & \textbf{-8.90} & \textbf{-8.81}  & \textbf{73.4\% } & \textbf{81.4\%} & 0.50 & 0.50 & 0.57 & 0.56 & 0.73 & 0.71 \\
Affinity+SA & -6.87 & -7.76 & -8.00 & -8.08 & -8.81 & -8.72 & 72.7\% & 80.8\% & \textbf{0.52} & \textbf{0.55 } & \textbf{0.60} & \textbf{0.59} & \textbf{0.74 }& \textbf{0.73 } \\
Affinity+QED & \textbf{-7.11} & \textbf{-8.02} & -8.01 & -7.99 & -8.17 & -8.72 & 73.7\% & 82.0\% & 0.51 & 0.52 & 0.57 & 0.57 & 0.73 & \textbf{0.73} \\
\bottomrule
\end{tabular}
}}
\end{table}

\begin{table}[!t]
\caption{Comparison of DPO and E$^2$PO with pretrained and supervised fine-tuned models. \method with DPO takes energy ranking, and with E$^2$PO uses exact energy for preference optimization.}

\label{table:exp:abalation}
\centering
{\resizebox{\textwidth}{!}{
\begin{tabular}{l|cc|cc|cc|cc|cc|cc|cc}
\toprule
\textbf{Methods} & \multicolumn{2}{c|}{\textbf{Vina Score (↓)}} & \multicolumn{2}{c|}{\textbf{Vina Min (↓)}} & \multicolumn{2}{c|}{\textbf{Vina Dock (↓)}} & \multicolumn{2}{c|}{\textbf{High Affinity(↑)}}   & \multicolumn{2}{c|}{\textbf{QED(↑)}}  & \multicolumn{2}{c|}{\textbf{SA(↑)}} & \multicolumn{2}{c}{\textbf{Diversity(↑)}} \\
& Avg. & Med. & Avg. & Med. & Avg. & Med. & Avg. & Med. & Avg. & Med. & Avg. & Med. & Avg. & Med.\\
\midrule
IPDiff & -6.42 & -7.01 & -7.45 & -7.48 & -8.57 & -8.51 & 69.5\% & 75.5\% & \textbf{0.52} & \textbf{0.53} & \textbf{0.61} & \textbf{0.59} & \textbf{0.74} & \textbf{0.73} \\
$\text{IPDiff}_{\textsc{sft}}$ & -6.53 & -6.62 & -7.27 & -7.09 & -8.14 & -8.09 & 67.5\% & 72.5\% & 0.48 & 0.48 & \textbf{0.61} & \textbf{0.59} & 0.72 & 0.69 \\
\textbf{\method-DPO} & \underline{-6.81}  & \underline{-7.62} & \underline{-7.75} &  \underline{-7.79} &  \underline{-8.58} &  \underline{-8.55 } & \underline{69.7\%}  & \underline{71.1\%} & 0.50 & 0.51 & 0.56 & 0.56  & \textbf{0.74} & \underline{0.72} \\
\textbf{\method-E$^2$PO} & \textbf{-7.07} & \textbf{-7.95} & \textbf{-8.09} & \textbf{-8.17} & \textbf{-8.90} & \textbf{-8.81}  & \textbf{73.4\%} & \textbf{81.4\%} & 0.50 & 0.50 & \underline{0.57} & \underline{0.56} & \underline{0.73 } & 0.71 \\

\bottomrule
\end{tabular}
}}
\end{table}

\textbf{Effect of reward objectives.} To further explore the potential of \method, we evaluate the effect of combining optimization objectives ($\rvr = \rvr_{\text{affinity}} + \rvr_{\text{SA}}$; $  \rvr = \rvr_{\text{affinity}} + \rvr_{\text{QED}}$)  and investigate whether such a combined reward function can lead to better molecular properties to counter the trade-off we discussed before. As shown in \Cref{table:exp:multiple}, the results indicate that finetuning solely with binding affinity apparently achieves better performance in terms of binding affinity metrics. However, \method-Affinity+SA generates compounds with better drug-likeness (QED) and synthetic accessibility (SA). Both models exhibit similar performance in terms of structural diversity. This suggests that while \method-Affinity is superior for binding affinity, incorporating synthetic accessibility considerations (Affinity + SA) results in compounds that are more drug-like and easier to synthesize, enabling more efficient multi-objective drug development. Moreover, \method-Affinity+QED achieves better binding affinity compared with \method-Affinity, while the improvement in QED is relatively minimal. Thus, balancing these objectives highlights the potential for overcoming trade-offs in molecular optimization.

\textbf{Comparison with Supervised Fine-Tuning.} Supervised Fine-Tuning (SFT) serves as an alternative method for generating molecules with user-defined optimization objectives. We select the top 50\% protein-ligand samples with higher quality in user-defined reward from the training dataset and fine-tune the baseline model with the same training and sampling setting. The results in \Cref{table:exp:abalation} show that SFT did not show improvement over the baseline, and \method demonstrates significantly superior results compared to SFT. 
 
\textbf{Effect of preference optimization methods.}
As discussed in \Cref{subsec:method-epo}, the original DPO objective is vulnerable to overfitting and we propose to avoid it with regularization by weighting preference losses with the user-defined rewards.
We compare the direct use of energy preference optimization by ranking molecule pairs (\method-DPO) and exact energy optimization with user-defined reward function (\method-E$^2$PO) in \Cref{table:exp:abalation}.
The results show that \method-E$^2$PO achieves superior performance over \method-DPO in binding affinity metrics (Vina Score, Vina Min, Vina Dock) while maintaining competitive scores in QED, SA and diversity. In terms of drug-likeness and structural diversity, \method-E$^2$PO performs competitively, indicating that while it prioritizes binding affinity, it still maintains favorable drug-like properties and diversity. This further supports our previous hypothesis regarding the trade-off between binding affinity and molecular properties. An additional ablation study on the effect of exact energy optimization is presented in \Cref{app:sec:ablation}.  

\begin{wraptable}{r}[5pt]{0.4\textwidth} 
\vskip -0.2in
\centering
\setlength{\tabcolsep}{4pt} 
\small
\caption{Finetuning TargetDiff with \method. \method-T denotes our finetuned model with the same reward objective on TargetDiff.}
\label{tab:targetdiff}
\begin{tabular}{l|cc|cc}
\toprule
\textbf{Metric} & \multicolumn{2}{c}{\textbf{TargetDiff}} & \multicolumn{2}{c}{\textbf{\method-T}} \\
& Avg. & Med. & Avg. & Med. \\
\midrule
Vina Score & -5.47& -6.30 & \textbf{-5.81} & \textbf{-6.51} \\
Vina Min & -6.64 & -6.83 & \textbf{-6.94} & \textbf{-7.01} \\
Vina Dock & -7.80 & -7.91 & \textbf{-7.92} & \textbf{-7.97} \\
QED    & 0.48 & 0.48   & \textbf{0.56} & \textbf{0.56} \\
SA     & 0.58 & 0.58   & \textbf{0.62} & \textbf{0.60} \\
Diversity  & 0.72& 0.71   & \textbf{0.74 }& \textbf{0.75} \\
\bottomrule
\end{tabular}
\end{wraptable}
\textbf{General applicability to ligand diffusion models.} We further justify the general applicability of the proposed approach by finetuning another model, TargetDiff~\citep{guan20233d}, with exact energy optimization (\method-T), As shown in \Cref{tab:targetdiff}, \method-T surpasses TargetDiff on all binding affinity and molecular properties, with a 6.2\%, 16.6\%,  6.9\%, 2.8\% increase in Avg.~Vina Score, QED, SA, and diversity, respectively. The results further justify that our approach is generally applicable to diffusion-based SBDD models. 
Notably, \method-T archives even better QED and SA compared with \method, which allows users to choose the model based on the specific purpose for molecular properties. Also, we notice the percentage of improvement of binding affinity from TargetDiff to \method-T is slightly lower than that from IPDiff to \method. This can be explained as preference optimization is more effective when the model distribution is more similar to the preference data distribution, and IPDiff is shown to fit CrossDocked data better than TargetDiff~\citep{huang2023protein}.

\begin{figure*}[!t]
    \centering
     \vspace{-10pt}
    \includegraphics[width=\linewidth]{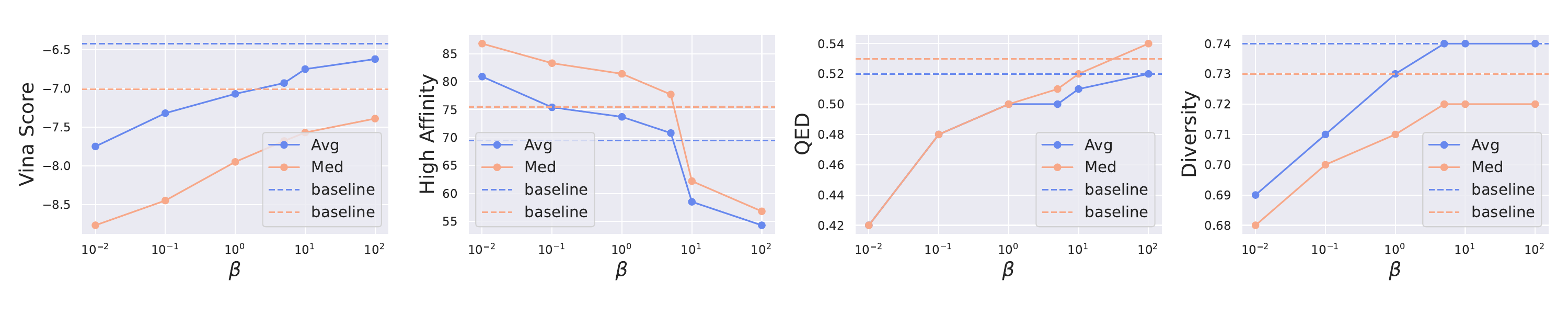}
     \vspace{-20pt}
    \caption{Ablation analysis of \method under different $\beta$. Vina Score, High Affinity, QED, and diversity are reported, where blue lines represent \method-DPO, and orange lines represent \method. The dotted lines represent the baseline IPDiff. }
    \label{fig:beta}
     \vspace{-5pt}
\end{figure*}




\textbf{Strength of $\beta$.}
We evaluate molecules generated by \method trained with varying $\beta$ values in \Cref{fig:beta}. Recall that $\beta$ influences the scale of energy preference optimization and regularization with respect to the reference model. The results indicate a clear trade-off between binding affinity and molecular properties with varying $\beta$. Lower $\beta$ values (e.g., 0.01) significantly enhance binding affinity metrics (Vina Score, Vina Min, Vina Dock), but at the cost of lower drug-likeness (QED) and diversity. Conversely, higher $\beta$ values improve QED, suggesting that these configurations generate more drug-like compounds while maintaining consistent synthetic accessibility and diversity.  We believe $\beta$ reaches an equilibrium around $\beta=1$, where binding affinity is maximized without sacrificing too much loss in molecular properties. This ablation study demonstrates that the parameter $\beta$ can offer a useful tool to train \method models with different desired trade-offs between binding affinity and useful molecular properties, which can vary for different drug development use cases.



\section{Conclusion}

In this paper, we present \method, a novel framework to align pretrained target-aware molecule diffusion models with desired functional properties via preference optimization. Our key innovation is the Exact Energy Preference Optimization method, which enables efficient and exact alignment of the diffusion model towards regions of lower binding energy and structural rationality specified by user-defined reward functions.
Extensive experiments on the CrossDocked2020 benchmark demonstrate the strong performance of \method. By incorporating user-defined reward functions and an improved Exact Energy Preference Optimization method, \method successfully achieves state-of-the-art performance in binding affinity while maintaining competitive molecular properties. 
In the future, we plan to explore more expressive molecular reward function classes within our framework and extend \method to real-world prospective drug design settings by integrating it into online drug discovery pipelines.


\section*{Acknowledgement}

We thank Jiaqi Han for the discussions on this project.
We gratefully acknowledge the support of ARO (W911NF-21-1-0125), ONR (N00014-23-1-2159), NVIDIA, and Chan Zuckerberg Biohub.
We also gratefully acknowledge the support of NSF under Nos. OAC-1835598 (CINES), CCF-1918940 (Expeditions), DMS-2327709 (IHBEM), IIS-2403318 (III); Stanford Data Applications Initiative, Wu Tsai Neurosciences Institute, Stanford Institute for Human-Centered AI, Chan Zuckerberg Initiative, Amazon, Genentech, GSK, Hitachi, SAP, and UCB.
Minkai Xu thanks the generous support of Sequoia Capital Stanford
Graduate Fellowship.



\bibliographystyle{plainnat}
\bibliography{ref}

\newpage
\appendix


\section{Limitations and Future Work}

While \method exhibits promising performance, there are still potential limitations to our current approach. For example, \method takes binding affinity as our reward function which is computed by AutoDock Vina~\citep{eberhardt2021autodock} in this work. However, computing binding energy via software is an approximation and sometimes can be very inaccurate. In the future, we plan to explore experiment-measured energy or ensemble different binding affinity calculation software, \eg, GlideScore~\citep{friesner2004glide}
In addition, in this work, we focus on an offline learning setting where the preference pairs are off-the-shelf. This is because computing binding affinity is computationally expensive. An important future direction to extend the work toward real-world drug discovery scenarios could be incorporating the online setting but with a limited number of query.

\section{Algorithm}
\label{app:sec:alg}


The pseudo-code for \method and \method-T are provided below. Sampling procedures are the same as \cite{guan20233d} and \cite{huang2023protein}.

\begin{algorithm}
\caption{Training Procedure \method}
\begin{algorithmic}[1]
\State \textbf{Input:} Protein-ligand binding dataset  $\{\gP, \gM^w, \gM^l \}_1^N$, pre-trained neural network $\phi_\theta$, reference network $\phi_\text{ref}$, learnable neural network $\psi_{\theta2}$ and pretrained interaction prior network $\psi_{\text{IP}}$.
\While{$\phi_\theta$ and $\psi_{\theta2}$ not converge}
    \State $[\![\rvp_{0}, \rvm_{0}^{w}, \rvm_{0}^{l}]\!] \sim \{ \gP, \gM^w, \gM^l \}^{N}_{i=1}$ where $\rvm_{0}^{w} = \{ \rvx_{0}^{w}, \rvv_{0}^{w}\}, \rvm_{0}^{l} = \{ \rvx_{0}^{l}, \rvv_{0}^{l}\}$
    \State Obtain $\rvr^w, \rvr^l$ for $\rvm^w, \rvm^l$, respectively.
    \State $t \sim U(0, \dots, T)$
    \State Move the complex to make CoM of protein atoms zero
    \State Obtain shifts  $[\![\rvs_0^{\gM_w}, \rvs_0^{\gM_l}]\!]$ and interactions $[\![\rvf_0^{\gM_{w}}, \rvf_0^{\gM_{l}}, \rvf_0^\gP]\!]$from $\psi_{\text{IP}}$ and $\psi_{\theta2}$ according to 
    \Statex \hspace{\algorithmicindent} ~\citep{huang2023protein}.
    \State Perturb $\rvx_0^{w}, \rvx_0^{l}$ to obtain $\rvx_t^{w}, \rvx_t^{l}$ with shifts  $\rvs_0^{\gM_w}, \rvs_0^{\gM_l}$
    \State \hspace{\algorithmicindent}  $\epsilon \sim \mathcal{N}(0, \textbf{I})$
    \State \hspace{\algorithmicindent} $\rvx_t^{w} = \sqrt{\bar{\alpha_t}} \rvx_0^{w} + \rvs_t^{\gM_w} + \sqrt{1 - \bar{\alpha_t}} \epsilon, \rvx_t^{l} = \sqrt{\bar{\alpha_t}} \rvx_0^{l} + \rvs_t^{\gM_l} + \sqrt{1 - \bar{\alpha_t}} \epsilon$
    \State Perturb $\rvv_0^{w},\rvv_0^{l}$ to obtain $\rvv_t^{w}, \rvv_t^{l}$
    \State \hspace{\algorithmicindent}    $g \sim \text{Gumbel}(0,1)$
    \State \hspace{\algorithmicindent}    $\log c^{w} = \log (\bar{\alpha_t} \rvv_0^{w} + (1 - \bar{\alpha_t}/K), \log c^{l} = \log (\bar{\alpha_t} \rvv_0^{l} + (1 - \bar{\alpha_t}/K)$
    \State \hspace{\algorithmicindent}    $\rvv_t^{w} = \text{onehot}(\arg \max_i (g_i + \log c_{i}^{w})), \rvv_t^{l} = \text{onehot}(\arg \max_i (g_i + \log c_{i}^{l}))$
    \State Embed $\rvv_t^{w}, \rvv_t^{l}$ into $\tilde{\rvh}_t^{w,0}, \tilde{\rvh}_t^{\gM_l,0}$, and embed $\rvv_0^{\gP}$ into $\tilde{\rvh}_t^{\gP,0}$
    \State Obtain features $[\![\rvh_t^{w,0}, \rvh_t^{\gM_l,0}, \rvh_t^{\gP,0}]\!]$ through prior-conditioning
    \State Predict $(\hat{\rvx}_{0|t}^{w}, \hat{\rvv}_{0|t}^{w})$ from $\phi_\theta([[\rvh_t^{\gM_w,0}, \rvh_t^{\gP,0}]], [\![\rvf_0^{\gM_{w}}, \rvf_0^\gP]\!])$
    \State Predict $(\hat{\rvx}_{0|t}^{l}, \hat{\rvv}_{0|t}^{l})$ from $\phi_\theta([\![\rvh_t^{\gM_l,0}, \rvh_t^{\gP,0}]\!], [\![\rvf_0^{\gM_{l}}, \rvf_0^\gP]\!])$
    \State Predict $(\hat{\rvx}_{0|t, \text{ref}}^{\text{w}}, \hat{\rvv}_{0|t, \text{ref}}^{\text{w}})$ from $\phi_\text{ref}([\![\rvh_t^{\gM_w,0}, \rvh_t^{\gP,0}]\!], [\![\rvf_0^{\gM_{w}}, \rvf_0^\gP]\!])$
    \State Predict $(\hat{\rvx}_{0|t, \text{ref}}^{l}, \hat{\rvv}_{0|t, \text{ref}}^{l}$ from $\phi_\text{ref}([\![\rvh_t^{\gM_{l},0}, \rvh_t^{\gP,0}]\!], [\![\rvf_0^{\gM_{l}}, \rvf_0^\gP]\!])$

    \State Compute loss $L$ with $(\hat{\rvx}_{0|t}^{w}, \hat{\rvv}_{0|t}^{w})$, $(\rvx_{0}^{l}, \rvv_{0}^{l})$, $(\hat{\rvx}_{0|t, \text{ref}}^{w}, \hat{\rvv}_{0|t, \text{ref}}^{w})$, $(\hat{\rvx}_{0|t, \text{ref}}^{l}, $and $\hat{\rvv}_{0|t, \text{ref}}^{l})$ 
    \Statex\hspace{\algorithmicindent} according to \Cref{eq:loss-alidiff-exact} 
   
    \State Update $\theta$ and $\theta2$ by minimizing $L$
\EndWhile
\end{algorithmic}
\end{algorithm}

\begin{algorithm}
\caption{Training Procedure for \method-T}
\begin{algorithmic}[1]
\State \textbf{Input:} Protein-ligand binding dataset  $\{\gP, \gM^w, \gM^l \}_1^N$, pre-trained neural network $\phi_\theta$, reference network $\phi_\text{ref}$
\While{$\phi_\theta$ not converge}
    \State $[\![\rvp, \rvm_{0}^{w}, \rvm_{0}^{l}]\!] \sim \{ \gP, \gM^w, \gM^l \}^{N}_{i=1}$ where $\rvm_{0}^{w} = \{ \rvx_{0}^{w}, \rvv_{0}^{w}\}, \rvm_{0}^{l} = \{ \rvx_{0}^{l}, \rvv_{0}^{l}\}$
    \State Sample diffusion time $t \sim U(0, \ldots, T)$
    \State Move the complex to make CoM of protein atoms zero
    \State Perturb $\rvx_0^w, \rvx_0^l$ to obtain $\rvx_t^w, \rvx_t^l$: $\rvx_t^w = \sqrt{\alpha_t \rvx_0^w + (1 - \alpha_t)\epsilon}$, $\rvx_t^l = \sqrt{\alpha_t \rvx_0^l + (1 - \alpha_t)\epsilon}$, 
    \Statex\hspace{\algorithmicindent} where $\epsilon \sim \mathcal{N}(0, I)$
    \State Perturb $\rvv_0^w, \rvv_0^l$ to obtain $v_t^w, v_t^l$: 
    \State \hspace{\algorithmicindent} $log c^w = log(\alpha_t \rvv_0^w + (1 - \alpha_t)/K)$
    \State \hspace{\algorithmicindent} $log c^l = log(\alpha_t \rvv_0^l + (1 - \alpha_t)/K)$
    \State \hspace{\algorithmicindent} $\rvv_t^w = \text{one\_hot}(\arg\max[g_i + log c_i^w])$
    \State \hspace{\algorithmicindent} $\rvv_t^l = \text{one\_hot}(\arg\max[g_i + log c_i^l]), $ where $g \sim \text{Gumbel}(0, 1)$
    \State Predict $[\hat{\rvx}_0^w, \hat{\rvv}_0^w]$ from $[\rvx_t^w, \rvv_t^w]$ with $\phi_\theta$: $[\hat{\rvx}_0^w, \hat{\rvv}_0^w] = \phi_\theta([\rvx_t^w, \rvv_t^w], t, \rvp)$
    \State Predict $[\hat{\rvx}_0^l, \hat{\rvv}_0^l]$ from $[\rvx_t^l, \rvv_t^l]$ with $\phi_\theta$: $[\hat{\rvx}_0^l, \hat{\rvv}_0^l] = \phi_\theta([\rvx_t^l, \rvv_t^l], t, \rvp)$
    \State Predict $[\hat{\rvx}_0^w, \bar{\rvv}_0^w]$ from $[\rvx_t^w, \rvv_t^w]$ with $\phi_\text{ref}$: $[\hat{\rvx}_0^w, \hat{\rvv}_0^w] = \phi_\text{ref}([\rvx_t^w, \rvv_t^w], t, \rvp)$
    \State Predict $[\bar{\rvx}_{0, \text{ref}}^l, \bar{\rvv}_{0, \text{ref}}^l]$ from $[\rvx_t^l, \rvv_t^l]$ with $\phi_\text{ref}$: $[\bar{\rvx}_{0, \text{ref}}^l, \bar{\rvv}_{0, \text{ref}}^l] = \phi_\text{ref}([\rvx_t^l, \rvv_t^l], t, \rvp)$
    \State Compute $\mathcal{L}(\theta) =  \mathcal{L}_\rvx(\theta) + \alpha \mathcal{L}_\rvv(\theta)$  according to \Cref{eq:loss-alidiff} 
    \State Update $\theta$ by minimizing $L$
\EndWhile
\end{algorithmic}
\end{algorithm}

\newpage
\section{Proof}
\label{app:sec:proof}

\newcommand{\KLDiv}{\mathbb{D}_{\textup{KL}}}

\begingroup
\def\thetheorem{\ref{theorem-kl}}
\begin{theorem}
    The objective function \Cref{eq:loss-alidiff-exact} optimizes a variational upper bound of the KL-divergence $\mathbb{D}_\textnormal{KL}\big(\hat{p}^*(\rvm|\rvp)||\hat{p}_\theta(\rvm|\rvp)\big)$, where $\hat{p}^*(\rvm|\rvp) \propto p_\textnormal{ref}(\rvm|\rvp) \exp(r(\rvm, \rvp))$ and $\hat{p}_\theta(\rvm|\rvp) \propto p_\textnormal{ref} (\rvm|\rvp) \left(\frac{p_\theta (\rvm|\rvp)}{p_\textnormal{ref}(\rvm|\rvp)}\right)^\beta $.
\end{theorem}
\addtocounter{theorem}{-1}
\endgroup

We prove the theorem with \Cref{app:lemma:bound,app:lemma:exact}. \Cref{app:lemma:bound} justifies the least square objective is the variational upper bound for preference optimization, and \Cref{app:lemma:exact} shows that regularized preference optimization corresponds to exact KL divergences between the optimal and parameterized distributions. 
A version of similar proof can be found in \citet{wallace2023diffusion} and \citet{ji2024towards,chen2024noise} respectively, and to be self-contained we incorporate these proofs here. Compared with \citet{wallace2023diffusion}, we introduce an additional term into the diffusion optimization. And compared with \citet{ji2024towards,chen2024noise}, we explicitly drop the assumption for drawing infinite samples $\rvm$ for each pocket $\rvp$.

\begin{lemma}
\label{app:lemma:bound}
    The objective function \Cref{eq:loss-alidiff-exact} $\mathcal{L}_\text{\method-E$^2$PO}(\theta) =  -\mathbb{E}_{(\rvp, \rvm_0^w, \rvm_0^l) \sim \mathcal{D}, t\sim[0,T], \rvm_{t}^w \sim q, \rvm_{t}^l\sim q} \big[ (\sigma (\rvr^w -\rvr^l)) (\gL_{t-1}^{\rvx} + \gL_{t-1}^{\rvv} ) + (1-\sigma (\rvr^w -\rvr^l)) (\bar{\gL}_{t-1}^{\rvx} + \bar{\gL}_{t-1}^{\rvv} ) \big]$ is a variational upper bound of:
    \begin{equation}
    \begin{aligned}
    \label{eq:loss-exact}
        \mathcal{L}_\text{E$^2$PO}(\theta) = & -\mathbb{E}_{(\rvp, \rvm^w, \rvm^l) \sim \mathcal{D}} \Big[ \big(\sigma (\rvr^w -\rvr^l)\big) \Big( \log \sigma \big( \beta \log \frac{p_\theta(\rvm^w | \rvp)}{p_{\text{ref}}(\rvm^w | \rvp)} - \beta \log \frac{p_\theta(\rvm^l | \rvp)}{p_{\text{ref}}(\rvm^l | \rvp)} \big) \Big) \\
        & + \big(1-\sigma (\rvr^w -\rvr^l)\big) \Big( \log \sigma \big( \beta \log \frac{p_\theta(\rvm^w | \rvp)}{p_{\text{ref}}(\rvm^w | \rvp)} - \beta \log \frac{p_\theta(\rvm^l | \rvp)}{p_{\text{ref}}(\rvm^l | \rvp)} \big) \Big) \Big].
    \end{aligned}
    \end{equation}
\end{lemma}
We refer readers to Appendix S2 of Diffusion-DPO~\citep{wallace2023diffusion} for the full proof. The bound is derived from Jensen's inequality and the convexity of the function $-\log \sigma$.

\begin{lemma}
\label{app:lemma:exact}
    The objective function \Cref{eq:loss-exact} optimizes the KL-divergence $\mathbb{D}_\textnormal{KL}\big(\hat{p}^*(\rvm|\rvp)||\hat{p}_\theta(\rvm|\rvp)\big)$, where $\hat{p}^*(\rvm|\rvp) \propto p_\textnormal{ref}(\rvm|\rvp) \exp(r(\rvm, \rvp))$ and $\hat{p}_\theta(\rvm|\rvp) \propto p_\textnormal{ref} (\rvm|\rvp) \left(\frac{p_\theta (\rvm|\rvp)}{p_\textnormal{ref}(\rvm|\rvp)}\right)^\beta $.
\end{lemma}

\begin{proof}
    First of all, we can rewrite the objective \Cref{eq:loss-exact} in the following form, expanding the $\operatorname{sigmoid}$ function:
    \begin{equation}
    \begin{aligned}
        \mathcal{L}_\text{E$^2$PO}(\theta) &=\mathbb{E}_{\rvp\sim\gD} \mathbb{E}_{  p_{\textnormal{ref}} (\rvm_{1:2}|\rvp)}\Bigg[ - \sum_{i=1}^2 \frac{e^{r(\rvp, \rvm_{i})}}{ \sum_{j=1}^2 e^{r(\rvp, \rvm_{j})}} \log \frac{
        e^{ \beta \log \frac{p_\theta(\rvm_{i}|\rvp)}{ p_{\textnormal{ref}} (\rvm_{i}|\rvp)} }}{\sum_{j=1}^2
        e^{ \beta \log \frac{p_\theta(\rvm_{j}|\rvp)}{ p_{\textnormal{ref}} (\rvm_{j}|\rvp)} } }\Bigg]\\
        &=\mathbb{E}_{\rvp\sim\gD} \mathbb{E}_{  p_{\textnormal{ref}} (\rvm_{1:2}|\rvp)}\Bigg[ - \sum_{i=1}^2 \frac{e^{r(\rvp, \rvm_{i})}}{ \sum_{j=1}^2 e^{r(\rvp, \rvm_{j})}} \log \frac{
        e^{ \log \big(\frac{p_\theta(\rvm_{i}|\rvp)}{ p_{\textnormal{ref}} (\rvm_{i}|\rvp)}\big) ^\beta }}{\sum_{j=1}^2
        e^{ \log \big(\frac{p_\theta(\rvm_{j}|\rvp)}{ p_{\textnormal{ref}} (\rvm_{j}|\rvp)}\big)^\beta } }\Bigg]\\
        &=\mathbb{E}_{\rvp\sim\gD} \mathbb{E}_{  p_{\textnormal{ref}} (\rvm_{1:2}|\rvp)}\Bigg[ - \sum_{i=1}^2 \frac{e^{r(\rvp, \rvm_{i})}}{ \sum_{j=1}^2 e^{r(\rvp, \rvm_{j})}} \log \frac{
        { \big(\frac{p_\theta(\rvm_{i}|\rvp)}{ p_{\textnormal{ref}} (\rvm_{i}|\rvp)}\big) ^\beta }}{\sum_{j=1}^2
        {  \big(\frac{p_\theta(\rvm_{j}|\rvp)}{ p_{\textnormal{ref}} (\rvm_{j}|\rvp)}\big)^\beta } }\Bigg]
    \end{aligned}
    \end{equation}
    By the definition $\hat{p}_\theta(\rvm|\rvp) \propto p_\textnormal{ref}^{1-\beta} (\rvm|\rvp) p_\theta^\beta (\rvm|\rvp)$, we have
    $\frac{\hat{p}_\theta(\rvm|\rvp)}{p_\textnormal{ref} (\rvm|\rvp)} \propto \Big(
        \frac{p_\theta (\rvm|\rvp)}{p_\textnormal{ref} (\rvm|\rvp)}
        \Big)^\beta
    $ (by dividing both sides with $p_\textnormal{ref} (\rvm|\rvp)$).
    Then we can substitute this equation and rewrite $\mathcal{L}_\text{E$^2$PO}(\theta)$:
    \begin{align}
    \mathcal{L}_\text{E$^2$PO}(\theta) &=\mathbb{E}_{\rvp\sim\gD} \mathbb{E}_{  p_{\textnormal{ref}} (\rvm_{1:2}|\rvp)}\Bigg[ - \sum_{i=1}^2 \frac{e^{r(\rvp, \rvm_{i})}}{ \sum_{j=1}^2 e^{r(\rvp, \rvm_{j})}} \log \frac{
    { \big(\frac{p_\theta(\rvm_{i}|\rvp)}{ p_{\textnormal{ref}} (\rvm_{i}|\rvp)}\big) ^\beta }}{\sum_{j=1}^2 { \big(\frac{p_\theta(\rvm_{j}|\rvp)}{ p_{\textnormal{ref}} (\rvm_{j}|\rvp)}\big)^\beta } }\Bigg]\nonumber\\
    &=\mathbb{E}_{\rvp\sim\gD}
    \mathbb{E}_{ p_{\textnormal{ref}} (\rvm_{1:2}|\rvp)}
    \Bigg[ 
    - \sum_{i=1}^2
    \frac{
        e^{r(\rvp, \rvm_{i})}
    }{  
        \sum_{j=1}^2 e^{r(\rvp, \rvm_{j})}
    }
    \log 
    \frac{
        { \frac{\hat{p}_\theta(\rvm_{i}|\rvp)}{ p_{\textnormal{ref}} (\rvm_{i}|\rvp)}
        }
    }{
        \sum_{j=1}^2
        {
            \frac{\hat{p}_\theta(\rvm_{j}|\rvp)}{ p_{\textnormal{ref}} (\rvm_{j}|\rvp)}
        }
    }
    \Bigg]
    \label{app:eq:dpo_intermediate}
    \end{align}
    Since $p_{\textnormal{ref}} (\cdot|\rvp)$ is supervised fine-tuned on samples $\{\rvm_i\}_{i=1}^2$, we can assume $\{\rvm_i\}_{i=1}^2$ takes most of the probability mass and thus $\mathbb{E}_{p_\textnormal{ref}(\rvm|\rvp)} \approx \mathbb{E}_{p_\textnormal{ref}(\rvm_{1:2}|\rvp)}$. Then we have the following approximation:
\begin{align*}
    & \sum_{j=1}^2 
    \frac{
                \hat{p}_\theta(\rvm_j|\rvp)    
            }{
                p_\textnormal{ref}(\rvm_j|\rvp)    
            }
    \approx 2 \mathbb{E}_{p_\textnormal{ref}(\rvm|\rvp)}
    \bigg[    
    \frac{\hat{p}_\theta(\rvm|\rvp)}
    {p_\textnormal{ref}(\rvm|\rvp)}
    \bigg] = 2\sum_{\rvm \in \gM} p_\textnormal{ref}(\rvm|\rvp)
    \frac{\hat{p}_\theta(\rvm|\rvp)}
    {p_\textnormal{ref}(\rvm|\rvp)} = 2\sum_{\rvm \in \gM} \hat{p}_\theta(\rvm|\rvp) =2,\\
    & \sum_{j=1}^2 e^{r(\rvp, \rvm_j)} \approx 2 \mathbb{E}_{p_\textnormal{ref}(\rvm|\rvp)}\Big[
    e^{r(\rvp, \rvm)}
    \Big] = 2 \sum_{\rvm\in\gM}p_\textnormal{ref}(\rvm|\rvp)
    e^{r(\rvp, \rvm)} =2 Z(\rvp).
\end{align*} 
Then we can plug the above results into \Cref{app:eq:dpo_intermediate} and further simplify $\mathcal{L}_\text{E$^2$PO}$:
\begin{align*}
    \mathcal{L}_\text{E$^2$PO}(\theta) 
    & = 
    \mathbb{E}_{\rvp\sim\gD}
    \mathbb{E}_{ p_{\textnormal{ref}} (\rvm_{1:2}|\rvp)}
    \Bigg[ 
    -\sum_{i=1}^2 
    \frac{e^{ r(\rvp, \rvm_i)}}{2 Z(\rvp)}
    \log 
    \frac{
        \hat{p}_\theta(\rvm_i|\rvp)
    }{
        2  p_{\textnormal{ref}} (\rvm_i|\rvp)    
    }
    \Bigg]\\
    & = 
    \mathbb{E}_{\rvp\sim\gD}
    \mathbb{E}_{ p_{\textnormal{ref}} (\rvm_{1:2}|\rvp)}
    \Bigg[ 
    -\sum_{i=1}^2 
    \frac{e^{ r(\rvp, \rvm_i)}}{2 Z(\rvp)}
    \log \bigg(
    \frac{
        \hat{p}_\theta(\rvm_i|\rvp)
    }{
        p_{\textnormal{ref}} (\rvm_i|\rvp) \frac{e^{ r(\rvp, \rvm_i)}}{ Z(\rvp)}   
    }
    \frac{e^{ r(\rvp, \rvm_i)}}{2 Z(\rvp)} \bigg)
    \Bigg]\\
    & = 
    \mathbb{E}_{\rvp\sim\gD}
    \mathbb{E}_{ p_{\textnormal{ref}} (\rvm_{1:2}|\rvp)}
    \Bigg[ 
    -\sum_{i=1}^2 
    \frac{e^{ r(\rvp, \rvm_i)}}{2 Z(\rvp)}
    \log \bigg(
    \frac{
        \hat{p}_\theta(\rvm_i|\rvp)
    }{
        p_{\textnormal{ref}} (\rvm_i|\rvp) \frac{e^{ r(\rvp, \rvm_i)}}{ Z(\rvp)}   
    } \bigg) 
    -\sum_{i=1}^2 
    \frac{e^{ r(\rvp, \rvm_i)}}{2 Z(\rvp)}
    \log
    \bigg( \frac{e^{ r(\rvp, \rvm_i)}}{2 Z(\rvp)} \bigg)
    \Bigg],
\end{align*}
where the second term remains constant $C$ to $\theta$, and thus can be omitted when analyzing the optimization for $\theta$. 
Notice the normalized form of $
\hat{p}^*(\rvm|\rvp) = 
\frac{1}{Z(\rvp)}
 p_{\textnormal{ref}} (\rvm|\rvp)
e^{ r(\rvp, \rvm)}
$, we replace $\frac{1}{Z(\rvp)}
 p_{\textnormal{ref}} (\rvm|\rvp)
e^{ r(\rvp, \rvm)}$ with $\hat{p}^*$ and further simplify the above equation:
\begin{align*}
    \mathcal{L}_\text{E$^2$PO}(\theta)
    &= \mathbb{E}_{\rvp\sim\gD}
    \Bigg[
    -\frac{1}{2}\sum_{i=1}^2
    \bigg[
        \frac{e^{ r(\rvp, \rvm_i)}}{Z(\rvp)}
    \log \frac{
        \hat{p}_\theta(\rvm_i|\rvp)
    }{
        \hat{p}^*(\rvm_i|\rvp)    
    }
    \bigg]
    +C
    \Bigg]\\
    &=
    \mathbb{E}_{\rvp\sim\gD}
    \Bigg[
        -\mathbb{E}_{ p_{\textnormal{ref}} (\rvm|\rvp)}\bigg[
        \frac{e^{ r(\rvp, \rvm)}}{Z(\rvp)}
        \log \frac{
            \hat{p}_\theta(\rvm|\rvp)
        }{
            \hat{p}^*(\rvm|\rvp)    
        }
        \bigg]
        +C
    \Bigg]\\
    &= \mathbb{E}_{\rvp\sim\gD}
    \Bigg[
        -\sum_{\rvm \in \gM}
         p_{\textnormal{ref}} (\rvm|\rvp)\frac{e^{ r(\rvp, \rvm)}}{Z(\rvp)}
        \log \frac{
            \hat{p}_\theta(\rvm|\rvp)
        }{
            \hat{p}^*(\rvm|\rvp)    
        }
    +C
    \Bigg]\\
    &=
    \mathbb{E}_{\rvp\sim\gD}
    \Bigg[
        -\sum_{\rvm \in \gM}
        \hat{p}^*(\rvm|\rvp)
        \log \frac{
            \hat{p}_\theta(\rvm|\rvp)
        }{
            \hat{p}^*(\rvm|\rvp)    
        }
    + C
    \Bigg]\\
    &=
    \mathbb{E}_{\rvp\sim\gD}
    \Big[
        \KLDiv(\hat{p}^*(\cdot|\rvp)\|\hat{p}_\theta(\cdot|\rvp))
    + C
    \Big],
\end{align*}
which completes the proof of \Cref{app:lemma:exact}.
\end{proof}

\section{Implementation Details}
\label{app:sec:implementation}

\textbf{Data.} Following ~\citep{guan20233d}, proteins and ligands are expressed with atom coordinates and a one-hot vector containing the atom types. For proteins, each atom type is represented by a one-hot vector covering 20 distinct amino acids. Ligand atoms are encoded using a one-hot vector that discriminates among several elements, specifically {H, C, N, O, F, P, S, Cl}. Additionally, a one-dimensional binary flag is incorporated to differentiate whether atoms are part of the protein or the ligand. We further apply two separate single-layer Multi-Layer Perceptrons (MLPs) to transform the input data into 128-dimensional latent spaces, providing a compact and informative representation for subsequent computational stages. 

\textbf{Preference Pair Generation.} For each synthetic molecule, we first locate its corresponding protein binding site and compute reward according to user-defined reward function for all synthetic molecules of the corresponding the binding site. We select a losing sample with lower reward and construct the preference. The selection process is detailed in \Cref{app:sec:ablation}.

\textbf{Architecture.} We follow the same architecture as IPDiff~\citep{huang2023protein}, which includes a learnable diffusion denoising model $\phi_{\theta1}$, learnable neural network $\phi_{\theta2}$ and pretrained interaction prior network \textsc{IPNet}. The architecture of all models used in our method is the same as IPDiff. 

\textbf{Pretraining Details.} Following existing work, we adopted the Adam optimizer with a learning rate of 0.001 and parameters $\beta$ values of (0.95, 0.999). The training was conducted with a batch size of 4 and a gradient norm clipping value of 8. To balance the losses for atom type and atom position, we applied a scaling factor $\lambda$ of 100 to the atom type loss. Additionally, we introduced Gaussian noise with a standard deviation of 0.1 to the protein atom coordinates as a form of data augmentation. Our parameterized diffusion denoising model, IPDiff was trained on a single NVIDIA A6000 GPU and achieved convergence within 200k steps.

\textbf{Training Details.} For finetuning, the pre-trained diffusion model is further fine-tuned via the gradient descent method Adam with init learning rate=5e-6, betas=(0.95,0.999). We keep other setting the same as pretraining. We use $\beta = 5$ in \Cref{eq:dpo}. We trained our model with one NVIDIA GeForce GTX A100 GPU, and it could converge within 30k steps.

\section{More Experimental Results}
\label{app:sec:ablation}

\textbf{Effect of diffusion steps.} 
In \cref{fig:time_vis}, we present a comprehensive ablation study examining the impact of diffusion steps on the optimization of molecular properties using our novel \method framework. The visualizations at the top of the figure showcase the progressive refinement of molecular structures across increasing diffusion steps ($t=200$ to $t=1000$). These images clearly illustrate how our model gradually enhances the molecular fitting within the target binding site, which is critical for improving drug efficacy. The plotted data below provides a quantitative analysis of QED, SA, Vina Dock across all test targets. Notably, both \method (P) and \method (R) demonstrate significant improvements in QED and SA scores as the number of diffusion steps increases and exhibit a notable decrease in Vina Dock. Particularly, \method-E$^2$PO model have shown better performance across all three metrics, with significant improvement on binding affinity across the diffusion steps.

\begin{figure*}[!t]
    \centering
    \includegraphics[width=0.99\linewidth]{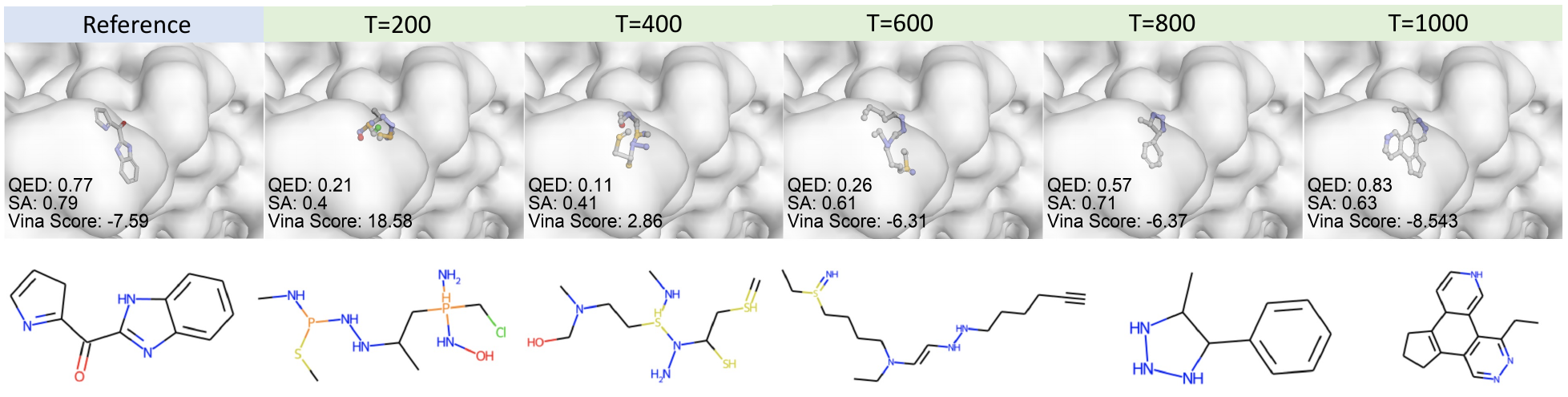}
     \vspace{-10pt}
     \includegraphics[width=0.99\linewidth]{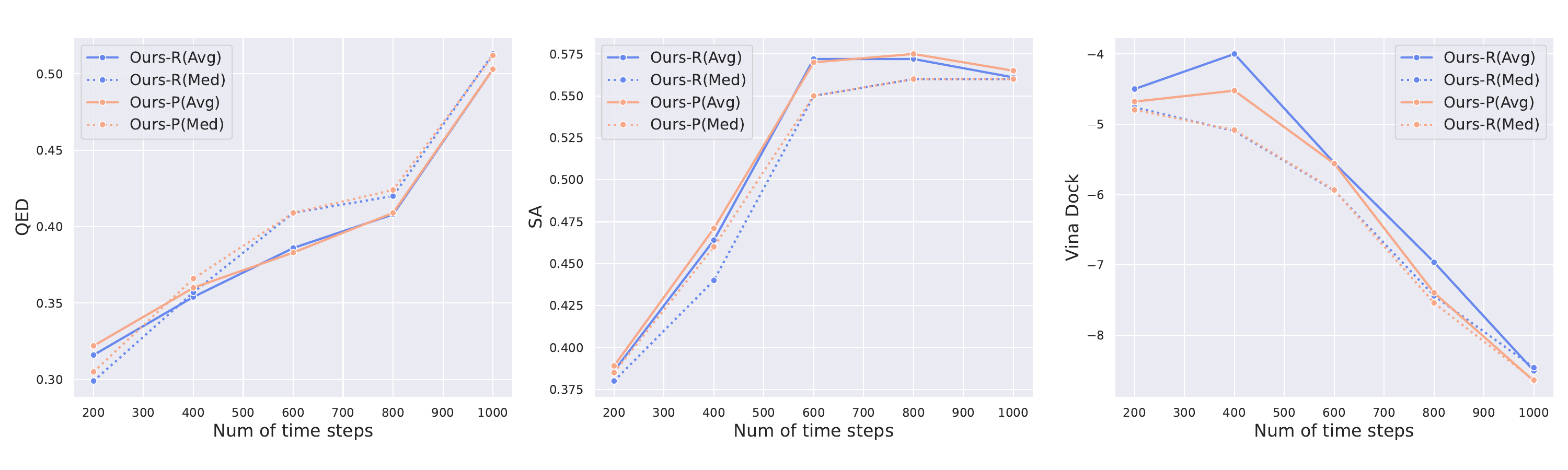}
    \caption{Ablation study on diffusion steps. The top shows a visualization of the generated molecule (4aua) under different time step. The bottom reports QED, SA and Vina Dock are reported under different diffusion steps(200, 400, 600, 800 and 1000). Blue lines represent \method-DPO and Red lines represent \method-E$^2$PO. }
    \vspace{-5pt}
    \label{fig:time_vis}
\end{figure*}


\begin{wraptable}{r}[5pt]{0.6\textwidth}
\vspace{-13pt}
\centering
\caption{Lipinski results for all methods.}
\label{table:exp:lip}
{\resizebox{0.99\linewidth}{!}{
\begin{tabular}{l|c|c|c|c|c|c}
\toprule
Methods & {\textbf{\method}} & \multicolumn{1}{c|}{\textbf{IPDiff}} & \multicolumn{1}{c|}{\textbf{TargetDiff}} & \multicolumn{1}{c|}{\textbf{AR}}   & \multicolumn{1}{c|}{\textbf{Pocket2Mol}} & \multicolumn{1}{c}{\textbf{Reference}}  \\
\midrule 
\textbf{Avg. Lipinski (↑)} & 4.48 & 4.52 & 4.51 & \underline{}{4.75} & \textbf{4.88} & 4.27\\
\bottomrule
\end{tabular}
}}
\vspace{-7pt}
\end{wraptable}

\textbf{Lipinski.}
We further compared Lipinski's Rule of Five~\citep{lipinski2012experimental} across all comparison methods. Lipinski's Rule of Five is another measurement for assessing drug-likeness besides QED, and we would like to incorporate this metric to validate our performance in generating drug-like molecules. The results of Lipinski's scores are reported in ~\Cref{table:exp:lip}. The results are consistent with our evaluation using QED score, as all diffusion-based models are not achieving high drug-likeness. We maintain similar drug-likeness as our backbone models targetDiff and IPDiff.

\begin{table}[!h]
\caption{Ablation study results with different choice of $\mathbf{m}^l$.}
\label{table:exp:ml}
\centering
{\resizebox{\textwidth}{!}{
\begin{tabular}{l|cc|cc|cc|cc|cc|cc|cc}
\toprule
\textbf{Choice of $\mathbf{m}^l$ } & \multicolumn{2}{c|}{\textbf{Vina Score (↓)}} & \multicolumn{2}{c|}{\textbf{Vina Min (↓)}} & \multicolumn{2}{c|}{\textbf{Vina Dock (↓)}} & \multicolumn{2}{c|}{\textbf{High Affinity(↑)}}   & \multicolumn{2}{c|}{\textbf{QED(↑)}}  & \multicolumn{2}{c|}{\textbf{SA(↑)}} & \multicolumn{2}{c}{\textbf{Diversity(↑)}} \\
& Avg. & Med. & Avg. & Med. & Avg. & Med. & Avg. & Med. & Avg. & Med. & Avg. & Med. & Avg. & Med.\\
\midrule 
worst & -7.07 & -7.95 & -8.09 & -8.17 & -8.90 & -8.81  & 73.4\% & 81.4\% & 0.50 & 0.50 & 0.57 & 0.56 & 0.73 & 0.71 \\
best & -6.80 & -7.66 & -7.83 & -7.69 & -8.64 & -8.05 & 70.2\% &  76.8\% & 0.50 & 0.52 & 0.56 & 0.55 & 0.74 & 0.71 \\
random & -6.96 & -7.82 & -8.03 & -8.00 & -8.77 & -8.20 & 72.1\% & 77.8\% & 0.50 & 0.51 & 0.56 & 0.55 & 0.74 & 0.72 \\
median & -6.96 & -7.85 & -8.01 & -7.96 & -8.80 & -8.24 & 72.5\% & 78.9\% & 0.50 & 0.51 & 0.57 & 0.55 & 0.74 & 0.72 \\ 
\bottomrule
\end{tabular}
}}
\end{table}

\textbf{Choice of $\mathbf{m}^l$.}  
Our generated dataset is obtained by directly transforming a standard labeled dataset into a pairwise preference dataset. Yet the binding affinity labels are continuous values where sometimes the difference between preferred and dispreferred is minimal. Therefore, the effect of energy preference optimization is highly sensitive to the overall data quality. \Cref{table:exp:ml} compares the performance of applying different strategies for selecting the dispreferred samples. "worst" indicates that the losing sample has the worst score from the user-defined reward function (lowest binding affinity). "best" suggests that the losing sample has the second-to-highest binding affinity(besides the preferred one). "random" and "median" mean that the losing samples are extracted randomly or from the median. Vina Score, Vina Min, Vina Dock, QED, SA, and Diversity are reported as average (Avg.) and median (Med.) values. Overall, the "worst" strategy, selecting the least favorable sample based on optimization objectiveness, consistently achieves the best performance in binding affinity metrics (Vina Score, Vina Min, and Vina Dock), while maintaining competitive drug-likeness (QED) and synthetic accessibility (SA). The "best" strategy, which may involve selecting the most favorable samples, performs poorly overall, which implies that energy preference optimization works better when there exists a larger discrepancy between $\mathbf{r}^w$ and $\mathbf{r}^l$. This allows the model to learn how to favor to $\mathbf{m}^w$ and avoid $\mathbf{m}^l$ during the finetuning process. The "random" and "median" strategies show intermediate performance, suggesting that a strategic approach to sample selection can significantly impact the efficacy of the resulting models.

\section{More Visualizations}
\label{app:sec:vis}

\begin{figure*}[!h]
    \centering
     \includegraphics[width=0.99\linewidth]{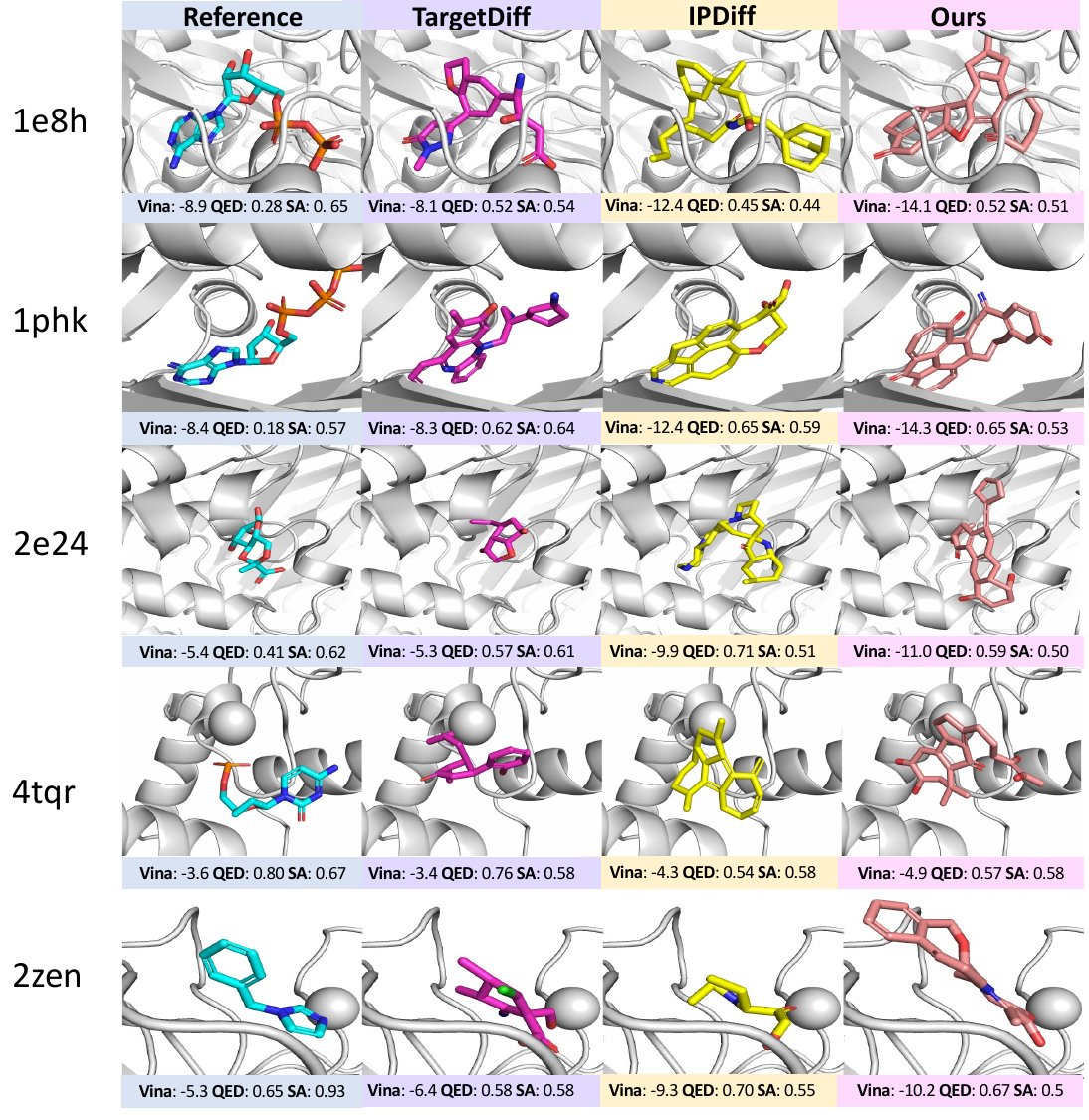} 
    \caption{More visualizations of generated ligands for protein pockets generated by TargetDiff, IPDiff, and \method.}
    \label{fig:more_vis}
\end{figure*}


\newpage
\section*{NeurIPS Paper Checklist}
\begin{enumerate}

\item {\bf Claims}
    \item[] Question: Do the main claims made in the abstract and introduction accurately reflect the paper's contributions and scope?
    \item[] Answer: \answerYes{} 
    \item[] Justification: Our claim is reflected through comprehensive experiments and theoretical proofs.

\item {\bf Limitations}
    \item[] Question: Does the paper discuss the limitations of the work performed by the authors?
    \item[] Answer: \answerYes{} 
    \item[] Justification: We have discussed limitations of our work.

\item {\bf Theory Assumptions and Proofs}
    \item[] Question: For each theoretical result, does the paper provide the full set of assumptions and a complete (and correct) proof?
    \item[] Answer: \answerYes{} 
    \item[] Justification: Yes we have provided proof in Appendix.

    \item {\bf Experimental Result Reproducibility}
    \item[] Question: Does the paper fully disclose all the information needed to reproduce the main experimental results of the paper to the extent that it affects the main claims and/or conclusions of the paper (regardless of whether the code and data are provided or not)?
    \item[] Answer: \answerYes{} 
    \item[] Justification: We have provided all the codes needed to reproduce the results presented. Pseudo code is also included in appendix. 

\item {\bf Open access to data and code}
    \item[] Question: Does the paper provide open access to the data and code, with sufficient instructions to faithfully reproduce the main experimental results, as described in supplemental material?
    \item[] Answer: \answerYes{} 
    \item[] Justification: We have provided all the codes needed to reproduce the results presented.

\item {\bf Experimental Setting/Details}
    \item[] Question: Does the paper specify all the training and test details (e.g., data splits, hyperparameters, how they were chosen, type of optimizer, etc.) necessary to understand the results?
    \item[] Answer: \answerYes{} 
    \item[] Justification: We have provided implementation details in appendix.

\item {\bf Experiment Statistical Significance}
    \item[] Question: Does the paper report error bars suitably and correctly defined or other appropriate information about the statistical significance of the experiments?
    \item[] Answer: \answerNo{} 
    \item[] Justification: Error bars are not reported because it would be too computationally expensive. Since our diffusion model sample 100 samples for each pocket, we believe reporting median and mean will be well reflected of the overall performance. 

\item {\bf Experiments Compute Resources}
    \item[] Question: For each experiment, does the paper provide sufficient information on the computer resources (type of compute workers, memory, time of execution) needed to reproduce the experiments?
    \item[] Answer: \answerYes{} 
    \item[] Justification: We have provided computer resources details in Appendix.
    
\item {\bf Code Of Ethics}
    \item[] Question: Does the research conducted in the paper conform, in every respect, with the NeurIPS Code of Ethics \url{https://neurips.cc/public/EthicsGuidelines}?
    \item[] Answer: \answerYes{} 
    \item[] Justification:  General machine learning method without specific concern in our mind.

\item {\bf Broader Impacts}
    \item[] Question: Does the paper discuss both potential positive societal impacts and negative societal impacts of the work performed?
    \item[] Answer: \answerNA{} 
    \item[] Justification: General machine learning method without specific concern in our mind.
    
\item {\bf Safeguards}
    \item[] Question: Does the paper describe safeguards that have been put in place for responsible release of data or models that have a high risk for misuse (e.g., pretrained language models, image generators, or scraped datasets)?
    \item[] Answer: \answerNA{} 
    \item[] Justification: General machine learning method without specific concern in our mind.

\item {\bf Licenses for existing assets}
    \item[] Question: Are the creators or original owners of assets (e.g., code, data, models), used in the paper, properly credited and are the license and terms of use explicitly mentioned and properly respected?
    \item[] Answer: \answerYes{} 
    \item[] Justification: All assets get credited. 

\item {\bf New Assets}
    \item[] Question: Are new assets introduced in the paper well documented and is the documentation provided alongside the assets?
    \item[] Answer: \answerYes{} 
    \item[] Justification: We have provided the code along with files to run the training process directly in the supplementary. 

\item {\bf Crowdsourcing and Research with Human Subjects}
    \item[] Question: For crowdsourcing experiments and research with human subjects, does the paper include the full text of instructions given to participants and screenshots, if applicable, as well as details about compensation (if any)? 
    \item[] Answer: \answerNA{} 
    \item[] Justification: No crowdsourcing and research with human subjects. 

\item {\bf Institutional Review Board (IRB) Approvals or Equivalent for Research with Human Subjects}
    \item[] Question: Does the paper describe potential risks incurred by study participants, whether such risks were disclosed to the subjects, and whether Institutional Review Board (IRB) approvals (or an equivalent approval/review based on the requirements of your country or institution) were obtained?
    \item[] Answer: \answerNA{} 
    \item[] Justification: No crowdsourcing and research with human subjects

\end{enumerate}

\end{document}